\newif\ifanonym
\anonymtrue
\anonymfalse 

\documentclass[11pt,letterpaper]{article}
\usepackage[margin=1in]{geometry}

\usepackage{amsmath,amsthm,amssymb,amsfonts}
\usepackage{xspace}
\usepackage{url}
\usepackage{hyperref}
\usepackage[dvipsnames]{xcolor}
\usepackage[ruled,linesnumbered,noend]{algorithm2e}

\hypersetup{colorlinks=true,allcolors=blue}

\usepackage{multirow}
\usepackage{longtable}
\usepackage{booktabs}
\usepackage{graphicx}
\usepackage{subcaption}
\usepackage{enumitem}

\usepackage{comment}

\usepackage[capitalise,noabbrev]{cleveref}
\usepackage{thm-restate}
\usepackage{thmtools}

\theoremstyle{plain}
\newtheorem{theorem}{Theorem}[section]
\newtheorem{lemma}[theorem]{Lemma}
\newtheorem{claim}[theorem]{Claim}
\newtheorem*{claim*}{Claim}
\newtheorem{corollary}[theorem]{Corollary}

\newtheorem{invariant}[theorem]{Invariant}

\theoremstyle{definition}
\newtheorem{definition}[theorem]{Definition}

\theoremstyle{remark}

\newtheorem*{remark*}{Remark}

\newcommand{\E}{\mathbb{E}}
\newcommand{\RR}{\mathbb{R}}

\newcommand{\benf}{\text{benf}}
\newcommand{\pes}{\text{pes}}
\newcommand{\tree}{\mathcal{T}}
\newcommand{\calE}{\mathcal{E}}

\DeclareMathOperator{\diam}{diam}
\DeclareMathOperator{\poly}{poly}

\DeclareMathOperator{\dist}{dist}
\DeclareMathOperator{\NN}{NN}
\DeclareMathOperator{\rep}{rep}
\DeclareMathOperator{\level}{lv}
\DeclareMathOperator{\buk}{buk}
\DeclareMathOperator{\buckets}{buks}

\DeclareMathOperator{\kcenter}{center}
\DeclareMathOperator{\kvalue}{value}
\DeclareMathOperator{\child}{child}
\DeclareMathOperator{\maxbenf}{maxbenf}

\usepackage[obeyFinal,textsize=footnotesize]{todonotes}
\newcommand{\shaofeng}[1]{\todo[color=blue!10,linecolor=red,inline]{\textbf{SJ:} #1}}

\title{Tree Embedding in High Dimensions: \\
Dynamic and Massively Parallel}
\ifanonym
\author{Anonymous Authors}
\else
\author{Gramoz Goranci\thanks{Faculty of Computer Science,
  University of Vienna, Austria (\texttt{gramoz.goranci@univie.ac.at}).} \and Shaofeng H.-C. Jiang\thanks{School of Computer Science, Peking University,
    China
   (\texttt{shaofeng.jiang@pku.edu.cn}).} \and Peter Kiss\thanks{Faculty of Computer Science,
  University of Vienna, Austria (\texttt{peter.kiss@univie.ac.at}). This research was funded in whole or in part by the Austrian Science Fund (FWF) 10.55776/ESP6088024.} \and Qihao Kong\thanks{School of Computer Science, Peking University,
  China (\texttt{2300012972@stu.pku.edu.cn}).} \and Yi Qian\thanks{School of Computer Science, Peking University,
  China (\texttt{qianyi@stu.pku.edu.cn}).} \and Eva Szilagyi\thanks{Faculty of Computer Science,
  UniVie Doctoral School Computer Science DoCS,
  University of Vienna, Austria
  (\texttt{eva.szilagyi@univie.ac.at}).}}
\fi
\date{}

\SetKwInOut{Input}{Input}\SetKwInOut{Output}{Output}

\begin{document}
\begin{titlepage}
        \maketitle
        \begin{abstract}

    Tree embedding has been a fundamental method in algorithm design with wide applications.
    We focus on the efficiency of building tree embedding in various computational settings under high-dimensional Euclidean $\mathbb{R}^d$.
    We devise a new tree embedding construction framework that operates on an arbitrary metric decomposition with bounded diameter,
    offering a tradeoff between distortion and the locality of its algorithmic steps.
    This framework works for general metric spaces and may be of independent interest beyond the Euclidean setting.
    Using this framework, we obtain a dynamic algorithm that maintains an $O_\epsilon(\log n)$-distortion tree embedding with update time $\tilde O(n^\epsilon + d)$ subject to point insertions/deletions,
    and a massively parallel algorithm that achieves $O_\epsilon(\log n)$-distortion in $O(1)$ rounds and total space $\tilde O(n^{1 + \epsilon})$ (for constant $\epsilon \in (0, 1)$).
    These new tree embedding results allow for a wide range of applications. Notably, under a similar performance guarantee as in our tree embedding algorithms, i.e., $\tilde O(n^\epsilon + d)$ update time and $O(1)$ rounds,
    we obtain $O_\epsilon(\log n)$-approximate dynamic and MPC algorithms for $k$-median and earth-mover distance in $\mathbb{R}^d$.

\end{abstract}         \thispagestyle{empty}
\end{titlepage}

\section{Introduction}
\label{sec:intro}

In the study of metric embeddings, the goal is to embed arbitrary metrics into simpler ones while incurring minimal distance distortion
A central concept in this area is that of \emph{probabilistic tree embeddings}, i.e., embedding any metric into a random tree metric with low expected distortion.
The seminal work of Bartal \cite{bartal1996probabilistic} gave the first competitive distortion bounds for such embeddings. Later, Fakcharoenphol, Rao, and Talwar (FRT) \cite{DBLP:journals/ejcss/FakcharoenpholRT04} showed that any finite metric space admits a probabilistic tree embedding with expected distortion that is logarithmic in the number of points, which is asymptotically optimal.
Probabilistic tree embeddings have proven fundamental in a wide range of applications in algorithm design and beyond, including clustering \cite{bartal2001approximating, cohen2021parallel, cohen2022scalable},
network design~\cite{awerbuch1997buy}, metric labeling~\cite{kleinberg2002approximation, chekuri2001approximation},
oblivious routing~\cite{wu2000polynomial}.

This very general tree embedding method is also useful for tackling computational challenges in the more specific Euclidean spaces $\mathbb{R}^d$,
especially when the dimension $d$ is general and can be part of the input.
In fact, the computational complexity of large $d$ differentiates from the low-dimensional $d = O(1)$ case;
moreover, many fundamental problems in metric spaces are already hard in $\mathbb{R}^d$ when $d = O(\log n)$.
For instance, in the context of polynomial time approximation algorithms, traveling salesperson problem (TSP)~\cite{trevisan2000hamming} and $k$-means~\cite{awasthi2015hardness} are known to be APX-hard when $d=\Omega(\log n)$.
Similarly, in the realm of fine-grained complexity, it has been shown that any $(1+\epsilon)$-approximation to diameter requires $\Omega(n^2)$ time~\cite{Williams18},
and in the streaming setting, any $(1+\epsilon)$-approximation for facility location~\cite{CJKVY22} and minimum spanning tree~\cite{chen2023streaming} requires $n^{\Omega(1)}$ space.

Indeed, tree embeddings have long served as a powerful tool for achieving competitive results for many high-dimensional optimization problems in different computational models.
In the offline setting, for example, Euclidean metrics admit a tree embedding algorithm that runs in $O(n^{1+\epsilon})$ time and achieves $O_{\epsilon}(\log n)$\footnote{$O_\epsilon(\cdot)$ hides small-degree  $\textrm{poly}(1/\epsilon)$ factor throughout (with degree $\geq 1$).} distortion\footnote{Although this result is not explicitly stated in the literature, it follows by combining known constructions of high-dimensional Euclidean spanners~\cite{Har-PeledIS13} with near-linear time algorithms for computing tree embeddings in graphs~\cite{MendelS09,blelloch2016efficient}.}.
More recently, the exploration of sub-linear computational models has gained significant momentum, and tree embeddings have continued to play a central role, leading to competitive or even state-of-the-art results.
Notably, fundamental problems such as facility location, minimum spanning tree (MST), geometric matching, and Earth Mover's Distance (EMD) have benefited from tree embedding constructions both in the streaming~\cite{Indyk04,DBLP:conf/stoc/ChenJLW22} and the massively parallel computation (MPC)~\cite{cohen2021parallel,AhanchiAHKZ23} settings.

Despite the versatility of tree embeddings in sub-linear models,
all abovementioned results (for sub-linear models) fall short of achieving the optimal distortion bound of $O(\log n)$. Moreover, beyond the streaming and MPC settings, it is not known if tree embedding in high dimensions can be maintained efficiently in the related, fully dynamic model, where the point sets evolve through insertions or deletions.
Technically, while the classical constructions~\cite{bartal1996probabilistic,DBLP:journals/ejcss/FakcharoenpholRT04} are known to be efficiently implementable in offline setting~\cite{MendelS09,blelloch2016efficient},
it still requires a new systematic framework for efficient tree embedding in various sub-linear models,
especially one that can utilize the special structure of the metric space, e.g., that in Euclidean $\mathbb{R}^d$, and beyond.

\subsection{Our Results}
\label{sec:results}

Our main contribution is a new tree embedding framework in general metrics that works with wide families of metric decompositions (which we discuss in more detail in \Cref{sec:tech_overview}).
This framework combining with geometric hashing techniques yields new results for tree embedding in dynamic and MPC settings,
which in turn yields rich applications.
Our approach may be of independent interest to tree embedding beyond Euclidean spaces.
Below, we start with introducing the dynamic result and its applications.

\paragraph{Dynamic Tree Embedding in $\mathbb{R}^d$.}
We obtain the first, distortion-optimal, fully dynamic algorithm for tree embeddings in high-dimensional Euclidean spaces.
For the sake of brevity, we assume that the aspect ratio of the underlying point set is bounded by $\poly(n)$ and that the dimension $d = O(\log n)$ (see Theorem~\ref{thm:dynamic} for a formal statement of our result). The exact guarantees of our result are summarized in the theorem below.

\begin{theorem}[Dynamic Tree Embedding; see \Cref{thm:dynamic}] \label{thm: dynamicFRT}
There is a fully dynamic algorithm such that for any $\epsilon \in (0,1)$, it supports insertions and deletions of points from Euclidean $\mathbb{R}^d$ and maintains a probabilistic tree embedding with $O_{\epsilon}(\log n)$ distortion in $\tilde{O}(n^{\epsilon})$ expected amortized update time.
\end{theorem}

For any constant $\epsilon$, our result achieves logarithmic distortion, which is asymptotically optimal~\cite{alon1995graph,bartal1998approximating}.\footnote{Our result also works if $\epsilon$ depends on $n$, for example $\epsilon = 1 / \log n$ (where the $\tilde{O}(n^\epsilon)$ time term becomes $\tilde{O}(1)$). However, in this parameter regime, our result is not competitive 
and it is open to find tight tradeoffs.
}
In contrast, two implicit approaches for dynamic tree embedding in $\mathbb{R}^d$
yield \emph{super-logarithmic} distortion. (i) Indyk's~\cite{Indyk04} tree embedding result achieves $O(d \log n)$ distortion and can handle point updates in $\textrm{poly} (\log n)$ update time.
(ii) The work of Forster, Goranci, and Henzinger~\cite{forster2021dynamic} maintains a tree embedding of dynamic graphs with distortion $O(\log^{3t-2} n)$ and update time $O(m^{1/t + o(1)} \log^{4t-2} n)$, for $t \geq 2$.
The high-dimensional spanner construction due to Har-Peled, Indyk and Sidiropulus~\cite{Har-PeledIS13} gives an $O(t)$-spanner with $\tilde{O}(n^{1+1/t^2})$ edges, which can be extended to support point updates in $\tilde{O}(n^{1/t^2})$ update, for $t \geq 1$. Combining these two algorithms, i.e., running the graph-based dynamic tree embedding result on the dynamic spanner, leads to sub-optimal distortion guarantees.   
Finally, we note that Jayaram, Waingarten and Zhang~\cite{DBLP:conf/stoc/JayaramWZ24} gave a close-to-optimal $\tilde O(\log n)$-distortion in $\poly\log(n)$ update time,
but the maintained tree is not 2-HST (see \Cref{sec:prelim} for a formal definition) which is structurally weaker and may not lead to the applications as we obtain in this paper.

\paragraph{Applications of Dynamic Tree Embedding.} 
As we previously highlighted, tree embeddings have led to solving many optimization problems across different computational models. The general recipe behind these applications is simple and elegant; instead of solving an optimization problem on arbitrary metrics, such embeddings allow us to solve the said problem on trees, which usually admits a much simpler solution, at the cost of paying the distance distortion on the quality of approximation. 

In a similar vein, our dynamic tree embedding from ~\Cref{thm: dynamicFRT} implies new dynamic algorithms for $k$-median clustering, Euclidean bipartite matching (also known as Earth Mover's Distance), general matching, and geometric transport. For the sake of brevity, we will only review the first two applications and refer the reader to ~\Cref{sec:application} for an extended discussion of all applications. 

In the Euclidean $k$-median clustering problem, the objective is to select a set of $k$ centers for a given set of $n$ points in $\mathbb{R}^{d}$ so as to minimize the total distance from each point to its nearest center. In the fully dynamic setting, the challenge is to maintain a good approximation (which is a set of $k$ centers) to the optimal clustering while efficiently supporting both insertions and deletions of points. Leveraging the structure of the algorithm from~\cite{cohen2021parallel} within our dynamic tree embedding \Cref{thm: dynamicFRT},
we obtain the following guarantees for fully dynamic $k$-median clustering.

\begin{corollary}[Dynamic $k$-median Clustering; see \Cref{thm:dyn_kmedian}]\label{thm:dynamicKmedian}
There is a fully dynamic algorithm such that for any $\epsilon \in (0,1)$, it supports insertions and deletions of points from Euclidean $\mathbb{R}^{d}$ and maintains an $O_\epsilon(\log n)$-approximation to the optimal $k$-median solution in $\tilde{O}(n^{\epsilon})$ expected amortized update time.  Furthermore, the center associated with any given point can be reported in $\tilde{O}(1)$ time.
\end{corollary}

Recently, considerable effort has been devoted to understanding the optimal trade-offs for $k$-median clustering in the dynamic setting~\cite{henzinger2020fully,fichtenberger2021consistent,TourHS24,BhattacharyaCGL24,BhattacharyaCF25}. Bhattacharya, Costa, and Farokhnejad~\cite{BhattacharyaCF25} achieve state-of-the-art guarantees, obtaining an $O(1)$-approximation with $\tilde{O}(k)$ update time.
However, this $\tilde O(k)$ is \emph{not} sublinear time in the worst case (i.e., $k = O(n)$).
In contrast, our algorithm always runs in sublinear small polynomial update time,
at the expense of incurring a logarithmic factor in the approximation ratio.
Previously, no known algorithm can achieve such sublinear update time with any $\poly\log n$ ratio.

We further consider the classic Euclidean bipartite matching problem in the dynamic setting. Although this problem has been recently studied in low dimensions~\cite{goranci2025fully}, the update time of their approach exhibits an inherent exponential dependency on the dimension. We use \Cref{thm: dynamicFRT} to give the first high-dimensional algorithm for this problem. 

\begin{corollary}[Dynamic Euclidean Bipartite Matching; see \Cref{thm:EMD}]\label{thm:dynamicEBM}
Let $A$, $B$, $|A|=|B| \leq n$ be sets of points in $\mathbb{R}^{d}$ that undergo insertions or deletions of pairs of points. There exists a dynamic algorithm such that for any $\epsilon \in (0,1)$, it maintains an expected $O_\epsilon(\log n)$ approximate minimum cost Euclidean matching of $A$ to $B$ in $\tilde{O}(n^{\epsilon})$ expected update time.  
\end{corollary}

\paragraph{MPC Tree Embeddings.}
Another major result of this paper is the first tight $O(\log n)$ distortion massively parallel tree embedding that runs in $O(1)$ rounds (under the typical setup of $n^{c}$ ($0 < c < 1$) local memory).

\begin{theorem}[MPC Tree Embedding; see \Cref{thm:mpc}]\label{thm: mpcFRT}
    There is an MPC algorithm such that given a set of $n$ points $P\subseteq \mathbb{R}^d$ distributed across machines with local memory $s \geq \poly\log(n) $ and a parameter $\epsilon \in (0, 1)$, it
    computes a tree embedding with $O_\epsilon(\log n)$ distortion
    in $O(\log_s n)$ rounds and $\tilde O(n^{1 + \epsilon})$ total space.
\end{theorem}

Our algorithm improves upon the MPC tree embedding algorithm of Ahanchi et al.~\cite{AhanchiAHKZ23} who achieve $O(\log^{1.5} n)$ distortion, still using $O(1)$ rounds, but a slightly worse $\tilde{O}(n^{1 + \epsilon})$ total space compared with their $\tilde{O}(n)$.
Combining with existing MPC approximation algorithms that are based on tree embedding, our new tree embedding readily yields $O(1)$-round, $O(\log n)$ approximate MPC algorithms for $k$-median, Euclidean MST (EMST), Earth Mover's Distance (EMD), and many more~\cite{cohen2021parallel,AhanchiAHKZ23}.
These results are the state-of-the-art (in ratio) for the mentioned problems,
in the fully scalable (which means the local space $s$ may be $n^c$ for arbitrary $0 < c < 1$ without dependence on any other parameters) $O(1)$-round regime:
previously, MPC algorithms for EMST either runs in $\omega(1)$ rounds~\cite{JayaramMNZ24,AzarmehrBJLMZ25}, or fully-scalable only for $d = O(1)$~\cite{DBLP:conf/stoc/AndoniNOY14},
and EMD and $k$-median admit $O(\log^{1.5} n)$-approximation in $O(1)$ rounds~\cite{AhanchiAHKZ23}.

\subsection{Technical Contributions}
\label{sec:tech_overview}

We provide a high-level picture of our technical contributions as follows.
Then in \Cref{sec:intro_ckr,sec:intro_imp_app} we give an overview of proofs, 
where we talk about the concrete technical challenges and our solution.

\paragraph{Tree Embedding via General Metric Decomposition.}
Although our results focus on Euclidean spaces,
the technical core actually works for \emph{general} metrics, which is our main technical contribution.
In particular, we devise a new tree embedding construction
operated on \emph{any} (given) \emph{bounded decomposition} (BD) for \emph{general} metrics.
We call a metric decomposition, which is a partition of the metric,
a $\tau$-bounded decomposition ($\tau$-BD) if every part has diameter $\tau$.
Our tree embedding offers a tradeoff between the distortion
and the \emph{locality} of its algorithmic steps,
controlled by an input parameter $\Gamma \geq 1$:
\begin{enumerate}
    \item \label{item:locality} (Locality) The algorithm operates via a subroutine
    that computes the intersections between radius-$\tau/ \Gamma$ metric balls with the parts of a $\tau$-BD (for some $\tau > 0$).
    \item (Distortion) The distortion is at most $O(\Gamma \log \Gamma) \cdot \log n$.
\end{enumerate}
Intuitively, this $\Gamma$ balances the tradeoff between the locality of the algorithm and the distortion: the larger the $\Gamma$, the more local the algorithm (which in turn is easier to implement in dynamic/MPC settings),
and the larger the distortion.
Note that this $\tau$-bounded guarantee is essential to many metric decomposition techniques,
such as padded decomposition and low-diameter decomposition (LDD)~\cite{KleinPR93,alon1995graph,bartal1996probabilistic,GuptaKL03,AbrahamBN06,elkin2018efficient,Filtser19,forster2019dynamic,AbrahamGGNT19,ConroyF25},
and our result applies to all of them.
In fact, the LDD techniques have been combined with a tree embedding construction from~\cite{bartal1996probabilistic}
to obtain efficient tree embedding in dynamic graph setting~\cite{forster2021dynamic},
and it is conceptually similar to ours.
However, their technique is more specific to LDD and it may not achieve the tight $O(\log n)$ distortion bound (if not even $\log^2 n$).

\paragraph{Efficient Implementations in Euclidean $\mathbb{R}^d$.}
To utilize the locality guarantee of \Cref{item:locality},
we employ a metric decomposition called sparse partitions~\cite{JiaLNRS05,Filtser24}.
Specifically, in Euclidean $\mathbb{R}^d$,
it serves as a $\tau$-BD such that any radius-$r/ \Gamma$ metric ball
intersects at most $n^{1 / \poly(\Gamma)}$ parts.
In our applications, we need to use consistent hashing~\cite{CJKVY22,filtser2025fasterapproximationalgorithmskcenter}
which is data-oblivious and time/space efficient construction of such sparse partition.
This enables us to precisely simulate our new CKR algorithm in the MPC setting.
We further show that the embedding undergoes $\tilde{O}(1)$ expected well structured changes under point insertions and deletions.
In turn, our new tree embedding results open a door to rich applications.
Notably, we devise new dynamic algorithms based on this new tree embedding maintaining \emph{optimal} solutions with respect to the dynamic embedding tree.

\begin{remark*}
    The fact that our new tree embedding framework works for general metrics may lead to results beyond Euclidean spaces,
    and this may be of independent interest.
    More specifically, our framework essentially reduces to finding efficient sparse partitions with competitive parameters.
    While efficient sparse partitions are better studied in Euclidean spaces,
    for many other spaces only existential tight bounds for sparse partitions are known~\cite{Filtser24} and efficient algorithms are yet to be devised. \end{remark*}

\subsubsection{Proof Overview: Tree Embedding via General Metric Decomposition}
\label{sec:intro_ckr}

\paragraph{Background: CKR Decomposition~\cite{CalinescuKR04}.}
Let $(V, \dist)$ be a metric space, let $B(x, r) := \{ y \in V : \dist(x, y) \leq r \}$ for $x \in V, r \geq 0$ be a metric ball,
and write $B^P(x, r) := B(x, r) \cap P$.
Existing constructions of tree embedding with tight $O(\log n)$ distortion (e.g.~\cite{DBLP:journals/ejcss/FakcharoenpholRT04,MendelS09,blelloch2016efficient})
are essentially based on the following key subroutine,
which is known as CKR decomposition~\cite{CalinescuKR04}.
The input is a dataset $P \subseteq V$ of $n$ points and a scale parameter $w$.
The CKR algorithm starts with imposing a random map $\pi : P \to [0, 1]$ 
(so with probability $1$ there is a one-to-one correspondence between $P$ and $\pi(P)$)
and uniform random $r \in [w / 4, w / 2]$,
then for each $p \in P$ computes its label $\ell_p := \arg\min\{ \pi(q) : q \in B^P(p, r)  \}$.
The labels implicitly define a partition of $P$, where the points with the same label goes to the same part.
The tree embedding can be implicitly defined by running this process 
independently for $i = 1, \ldots , m:=O(\log n)$ (assuming the aspect ratio is $\poly(n)$), using $w_i := 2^{m-i}$ for each distance scale $2^{m-i}$ to obtain the label $\ell_p^{(i)}$'s,
and the labels define the tree embedding.

\paragraph{Approximating CKR via Metric Decomposition.}
However, it is generally difficult to find the min $\pi$ value in a metric ball as in CKR, beyond brute-force enumeration.
One natural way to bypass this computational issue is to approximate the metric ball by a number of small enough pieces,
which may be described by a bounded decomposition.
Here, a $\tau$-bounded decomposition ($\tau$-BD) is a partition of $V$,
such that any part has diameter at most $\tau$.
Our algorithm is a direct generalization of the CKR decomposition~\cite{CalinescuKR04} via general BD.

\paragraph{Interpreting Metric Decomposition as Hashing.}
Next, we overview our modified CKR algorithm and its analysis.
We assume an (arbitrary) family of BD of the space is given,
and provided to us as a hash function $\varphi : V \to U$ (for some domain $U$),
such that $x, y\in V$ belong to the same part if and only if $\varphi(x) = \varphi(y)$.
This language of hashing (instead of partition) is more natural to us
because to implement the algorithm in Euclidean space and dynamic/MPC settings,
we need a data-oblivious decomposition for the entire $\mathbb{R}^d$,
and the efficient access of this is naturally provided via a geometric hashing.
For $x \in V$, let $\buk(x) := \varphi^{-1}(\varphi(x))$ be the bucket/part that contains $x$, and $\buk(x)^P := \buk(x) \cap P$.
For a set $S \subseteq V$ let $\buckets(S) := \bigcup_{x \in S} \buk(x)$.
We treat a single point $x$ as a singleton set so we can write $\buckets(x)$,
and since $\buckets(x) = \buk(x)$ we only use $\buckets(x)$ (and not $\buk(x)$) in the following discussion.
Write $\buckets^P(\cdot) := \buckets(\cdot) \cap P$.

\paragraph{Our Modified CKR Algorithm and Tree Embedding.}
In our modified CKR, apart from the input dataset $P \subseteq V$ and scale $w > 0$,
our algorithm additionally takes a parameter $\Gamma > 0$ as input.
As mentioned, this parameter balances the accuracy with the locality.
Our algorithm start with imposing an $O(w)$-BD/hashing.
Then a natural change we make to original CKR,
is to replace finding the minimum $\pi$ value in $B^P(p, r)$,
with that in the intersecting buckets $\tilde B(p, r) := \buckets^P(B(p, r))$,
i.e., $\ell_p := \pi_{\min}(\tilde B^P(p, r)) $, where $\pi_{\min}(S) := \min_{x \in S} \pi(x)$.
The other related change is that $r$ is picked uniformly at random from $[w / \Gamma / 4, w / \Gamma / 2]$,
scaled by $\Gamma$.
Finally, there is also a subtle change to defining tree embedding:
for each distance scale $2^{m-i}$, we run the modified CKR not only using the scale $w_i = 2^{m-i}$,
but also a new instance of $O(w_i)$-BD (over $i$)
which may not have any relation to other $2^j$-BD's for $j \neq i$,
especially that they are not nested/hierarchical.
This means we must also define $\tilde B^P$ as $\tilde B^P_i$ for scale $w_i$, instead of using a universal BD.
See \Cref{alg:frt} for the full description of this process.

\paragraph{Distortion Analysis: Review of~\cite{DBLP:journals/ejcss/FakcharoenpholRT04} Analysis.}
We focus on the distortion analysis, and (for now) ignore how the new algorithm is implemented.
We follow the high-level proof strategy as in~\cite{DBLP:journals/ejcss/FakcharoenpholRT04},
where a key step is the following:
fix $p, q \in P$, let $i^*$ be the smallest level satisfying $\ell_p^{(i^*)} \neq \ell_p^{(i^*)}$,
and we need to show $\E[2^{m-i^*}] \leq O(\log n) \cdot \dist(p, q)$.
Here, this $2^{m-i^*}$ is within $O(1)$ to the tree distance between $p$ and $q$,
and hence it bounds the distortion from above.
The distortion lower bound (sometimes called dominating property) is easy to show and we omit discussing it.
Intuitively, the event $\ell^{(i)}_p \neq \ell^{(i)}_q$ (for some $i$) means $p, q$ belong to different clusters at scale $2^{m-i}$,
and it makes $p$ and $q$ lie in different subtrees in the final tree embedding,
incurring $2^{m-i}$ cost, which we need to bound.
In FRT's analysis, one lists points in $B^P(p, r)$ in non-decreasing order of distance to $p$,
and write $p_j$ be the $j$-th one.
Then by the CKR process, for every $j$, $\ell^{(i)}_p$ (at any distance scale $w_i$) is realized by $p_j$ only if $p_1, \ldots, p_{j - 1}$ all have $\pi$ value larger than that of $p_j$'s.
Denote this event as $\mathcal{E}_j$, and this event has probability at most $1 / j$ where we (only) use the randomness of $\pi$:
\begin{equation}
    \label{eqn:intro_prob_Ei}
    \Pr[\mathcal{E}_j] \leq 1 / j.
\end{equation}
Then it is shown that the conditional expectation $\E[2^{m-i^*} \mid \mathcal{E}_j] \leq O(\dist(x, y))$ for any $j$,
which roughly follows from the fact that $2^{m-i^*}$ must be around $\Theta(\dist(x, y))$,
and from the randomness of $r$.
Taking the expectation over $\mathcal{E}_j$'s, one concludes
$\E[2^{m-i^*}] \leq \sum_{j} 1 / j \cdot O(\dist(x, y)) = O(\log n) \cdot \dist(x, y)$.

\paragraph{Distortion Analysis: Technical Challenges.}
To adapt this FRT's analysis in our case,
a major technical challenge is caused by using $\tilde B_i(\cdot, \cdot)$
to define the labels for each scale $w_i$ (instead of using the metric ball $B(\cdot, \cdot)$).
This approximate metric ball $\tilde B_i(p, r)$ is structurally sophisticated
and does not share nice geometric properties as in an exact metric ball,
especially considering we are working with an arbitrary bounded decomposition.
To see this, given that a point $q$ belongs to $\tilde B_i(p, r)$,
it does not imply a closer point $q'$ (to $p$) is also in $\tilde B_i(p, r)$,
and this breaks the probability $\Pr[\mathcal{E}_j] \leq 1 / j$ as in \eqref{eqn:intro_prob_Ei}.
We emphasize that these challenges are independent to the new parameter $\Gamma$,
and the technical challenges remain even if $\Gamma = 1$. 

\paragraph{New Insight: Representative Sets.}
To resolve this issue, our main insight is that it is possible to define
a set of \emph{representative} points $\rep(p, r) \subseteq V$ for $\tilde B(p, r)$ (and $\rep(p, r)$ may not be a subset of dataset $P$),
defined such that the algorithm may be roughly interpreted as running on $\rep(p, r)$,
and that it behaves similarly enough to the metric ball $B(p, r)$.
This holds for any metric decomposition.
In particular, it satisfies the following (see \Cref{def:rep}) and the existence of such representatives is shown in \Cref{lemma:rep_exist}, which is our major technical insight.
    \begin{enumerate}
        \item 
        (distinct)
        $\forall x \neq y \in \rep_i(p, r)$, $\buckets_i(x) \cap \buckets_i(y) = \emptyset$;
\item
                (monotone) $\forall r' \in (0, r)$,
        $\rep_i(p, r') \subseteq \rep_i(p, r)$;
        \item
        (ball-preserving)
        $\rep_i(p, r) \subseteq B(p, r)\cap\buckets_i(P)$ and $\buckets_i^P(\rep_i(p, r)) = \buckets_i^P(B(p, r))$.
\end{enumerate}
In this way, we can use $\rep_i^P(p, r)$ to replace $B^P(p, r)$ when computing the label $\ell^{(i)}_p$;
specifically, the label $\ell^{(i)}_p$ can be equivalently found only in $\rep_i(p, r)$,
instead of the whole $B(p, r)$, i.e., there is $\rep^*_i(p, r) \in \rep_i(p, r)$ 
such that $\ell^{(i)}_p = \pi_{\min}(\buckets^P_i(\rep^*_i(p, r)))$ (see \Cref{def:rep_star}).
Moreover, we can show an analogy to \eqref{eqn:intro_prob_Ei}:
for some scale $w_i$,
list $\rep_i(p, r)$ in non-decreasing order of distance to $p$ as $p_j$'s,
\begin{equation}
    \label{eqn:intro_pr_ub}
    \forall j, \quad \Pr[\pi_{\min} (\buckets_i^P(p_j)) = \ell^{(i)}_p] \leq H_{m_j} - H_{m_{j - 1}},
\end{equation}
where $H_t := \sum_{s=1}^t 1 / s$ is the harmonic sum,
and $m_t := |\bigcup_{s\leq t} \buckets^P(p_{s})|$ is the size of the first $t$ buckets (see \eqref{eqn:first_term} in \Cref{sec:proof_claim_bound_single_pair}).

In the existence proof of $\rep(\cdot, \cdot)$ (in \Cref{lemma:rep_exist}), we actually give an explicit definition of it (used only in the analysis). This definition is technical, and here we discuss the intuition.
A natural attempt to define $\rep(\cdot, \cdot)$ such that it satisfies all three of distinct, monotone and ball-preserving properties,
is to let $\rep_i(r, p) := B(p, r) \cap \buckets_i(P)$, which is precisely the first half of the ``ball-preserving'' bullet.
The issue is that it does not satisfy the distinctness (bullet 1).
Our definition fixes this by only selecting a subset of $B(p, r) \cap \buckets_i(P)$,
roughly the nearest neighbors from $p$ in each bucket of $B(p, r) \cap \buckets_i(P)$.
This ``nearest'' property ensures the distinctness (bullet 1).
Moreover, it is important to select the nearest neighbor
and it does not work if e.g., selecting a point in the buckets with smallest ID.
The reason is that nearest neighbor has additional geometric properties that ensure bullet 2 and 3 hold.

In fact, we may use a more straightforward statement of \eqref{eqn:intro_pr_ub}: list the buckets that intersect $B(p, r)$ in close-to-far order from $p$,
and
if we denote the $j$-th bucket (in this order) as $S_j$,
then $\Pr[\ell_p^{(i)} = \pi_{\min}(S_j)] \leq H_{m_j} - H_{m_j - 1}$.
However, we still need this explicit identification of representative points, as well as the abstract properties listed in the three bullets,
in proving \eqref{eqn:intro_pr_ub}, and the other parts of the analysis where the argument is inherently designed to work with points instead of buckets (see e.g., \eqref{eqn:first_step}, \eqref {eqn:split_probability}).

\paragraph{Additional Challenge: $\rep$ Set Depends on $i$.}
Unfortunately, \eqref{eqn:intro_pr_ub} cannot substitute \eqref{eqn:intro_prob_Ei} in the analysis,
and the key issue is that we must first fix $i$ before we can talk about the subset $\rep_i$ and the ordering of points,
whereas in \eqref{eqn:intro_prob_Ei} the point set as well as the ordering is fixed universally by the metric ball, independent of $i$.
This makes it impossible to use FRT's plan since the condition depends on $i$,
and it requires a union bound over the choice of $i$, which might introduce an additional $O(\log n)$ factor in the distortion (because there are $O(\log n)$ $i$'s).

\paragraph{Alternative Analysis Plan and Improved Bound.}
Our final analysis deviates from that of FRT,
and we indeed need a sum over $i$'s.
To avoid the additional $O(\log n)$ factor,
we derive a refined upper bound with respect to the scale $w_i$,
instead of a universal one.
The more difficult case is when $w_i$ is much larger than $\dist(p, q)$,
since if $\ell^{(i)}_p \neq \ell^{(i)}_q$ happens it incurs a large $w_i$ as error,
which can be sensitive.
For this case, we have an improved bound $O(\dist(p, q)) \cdot (H_{\tilde B_i(p, w_i)} - H_{\tilde B_{i}(p, w_i / 4)})$ (see \Cref{claim:bound_single_pair}).

A natural next step is to cancel the terms $H$ out via a telescoping sum when summing over $i$.
However, this cancellation requires us to relate $H_{\tilde B_i}$ with $H_{\tilde B_j}$ for $j \neq i$.
Although the decompositions in different scales may not be related in any way,
we can still use the diameter bound and triangle inequality to show that,
when the scales differ by a $\poly(\Gamma)$ factor,
the smaller-scale approximate ball $\tilde B^P$ would be a subset of the larger one, regardless of how the decomposition behaves; namely,
$\tilde B^P_{i + s}(p, w_i / 2^{s}) \subseteq \tilde B^P_{j}(p, r_i)$ for any $i, j$ and $s = O(\log \Gamma)$ (see \Cref{fact:subset_without_nested}).
This helps us to only incur an additional factor of $O(\Gamma \log \Gamma)$ instead of $O(\log n)$.

\subsubsection{Proof Overview: Implementations in $\mathbb{R}^d$ and Applications}
\label{sec:intro_imp_app}

\paragraph{Consistent Hashing.}
We employ consistent hashing~\cite{CJKVY22,filtser2025fasterapproximationalgorithmskcenter} techniques to utilize the locality property of our tree embedding algorithm in $\mathbb{R}^d$ (see \Cref{def:consistent_hashing}).
Roughly speaking, such hashing $\varphi : \mathbb{R}^d \to U$
yields data-oblivious $\tau$-bounded partitions (for any $\tau$) such that for any $B(p, \tau / \Gamma)$
it holds that $|\varphi(B(p, \tau / \Gamma))| \leq \Lambda := n^{\poly(1 / \Gamma)}$ for any $p \in \mathbb{R}^d$.
This exactly means $|\tilde B(p, r)| \leq \Lambda = n^{\poly(1 / \Gamma)}$ as in our tree embedding algorithm. Crucially, for the purposes of our dynamic implementation, the hash function of \cite{CJKVY22,filtser2025fasterapproximationalgorithmskcenter} allows for the computation of the value $\varphi(p)$ and set $\varphi(B(p, \tau / \Gamma))$ in $\poly(d)$ and $|\varphi(B(p, \tau / \Gamma))| \cdot \poly(d)$ time respectively.

\paragraph{Dynamic and MPC Implementations.}
The MPC implementation follows immediately by utilizing a low-space construction of consistent hashing~\cite{CJKVY22} (see \Cref{lemma:hashing}).
In the dynamic setting, we need to plug in a time-efficient consistent hashing bound, albeit the parameter is slightly worse (see \Cref{lemma:hashing_dynamic}).
However, the dynamic implementation requires even more work.
The first observation is that deletions can be ``ignored'',
since deleting a data point may be interpreted as deactivating a leaf node in an existing embedding tree, and doing so does not hurt the distortion.
Now suppose we work with insertion-only setting.
Consider an insertion of a point $p$ and some scale $w_i$.
We first insert $p$ to its bucket under $\varphi_i$ and maintain the min $\pi$-value of the bucket.
The bucket $\buckets(p)$ may intersect many $B(q, r_i)$'s (over $q \in P$),
and we need to scan and update the min $\pi$-value for all such $q$.
To do this efficiently, for each $q$ we maintain a priority queue for the intersecting buckets $\varphi_i(B(q, r_i))$,
with respect to the (min) $\pi$-value of the bucket,
so the total scan caused by inserting $p$ is proportional to the change of the min element of the priority queues.
The total/amortized update time can be bounded using two facts: a) each $|\varphi_i(B(q, r_i))| \leq \Lambda$ which follows from consistent hashing; and
b) the change of min for all priority queues combined is small,
which is argued via a low-recourse fact of permutations/random maps, also employed in~\cite{MendelS09}:
let $\pi$ be a permutation of $[n]$, then the local minimum $\E[\{ i : \pi(i) = \min\{ \pi(j) : j < i \} \}] = O(\log n)$.

\paragraph{Applications.}
The high-level idea of our applications is to utilize the simple structure of tree embedding, i.e., HST,
to find/maintain optimal solutions on the tree.
The applications in MPC setting is relatively straightforward, since previous results~\cite{cohen2021parallel, AhanchiAHKZ23} works with a general tree embedding in MPC in a black-box way.
Our dynamic applications generally do not work in a black-box way, and we need to introduce new steps to maintain the optimal solution on a tree.
We crucially exploit that the dynamic tree embedding maintained by our framework only undergoes $\tilde{O}(1)$ simple changes per updates to the underlying input points. Namely, the embedding itself may only undergo leaf deletions and the insertion of leafs with a path connecting them to existing nodes (for a precise description see \cref{thm:dynamic}).

Our dynamic geometric matching related results (Euclidean bipartite and general matching, geometric transport) all rely on showing that the greedy matching algorithms seamlessly extend to the dynamic model under these updates. Our implementation for maintaining a $k$-median solution is more intricate: we show that a different characterization of the algorithm of \cite{cohen2021parallel} for the problem is locally robust under the same set of updates.

     \subsection{Related Work}
\label{sec:related}

A pioneering work of Alon et al. \cite{alon1995graph}, has shown how to embed arbitrary metric spaces into tree metrics with sublinear distortion of $2^{O(\sqrt{\log n \log \log n})}$. In contrast, it is known that using deterministic embeddings even for simple graph metrics, such as a cycle graph, into a tree metric incurs a distortion of $\Omega(n)$ \cite{rabinovich1998lower}. Bartal \cite{bartal1996probabilistic, bartal1998approximating} improved the distortion to $O(\log ^2 n)$ and also established a lower bound of $O(\log n)$ for probabilistic methods. 
Fakcharoenphol et al. \cite{DBLP:journals/ejcss/FakcharoenpholRT04} proposed a general scheme that matches this lower bound for arbitrary metrics. Since then, a long line of work has focused on designing efficient algorithms within this framework to compute tree embeddings for specific metrics and computational settings.
For shortest-path metric of graphs, \cite{MendelS09,blelloch2016efficient} presented algorithms that computes such an embedding in $\tilde O(m)$ time,
and this combining with a spanner~\cite{Har-PeledIS13} readily implies an $O_\epsilon(\log n)$-distortion $\tilde O(nd + n^{1 + \epsilon})$-time construction in Euclidean $\mathbb{R}^d$.
Forster et al. \cite{forster2021dynamic} designed a dynamic algorithm that maintains a tree embedding $n^{o(1)}$ distortion and $n^{o(1)}$ update time for graph metrics.
Bartal et al. \cite{barta2020online} has shown a metric-oblivious algorithm that maintains a tree embedding with distortion $O(\log n \log \Delta)$ for general metric spaces in the online setting, where $\Delta >0$ denotes the aspect ratio of the input points.
The problem was also studied in parallel/distributed settings~\cite{khan2008efficient,blelloch2012parallel,GhaffariL14,friedrichs2018parallel,becker2024decentralized}.

 \section{Preliminaries}
\label{sec:prelim}

\paragraph{Notations.}
For integer $n$, let $[n]:=\{1,\ldots,n\}$, and let $H_n := \sum_{i \in [n]} 1 / i$ be the harmonic sum.
For a function $\varphi: X\to Y$, the image of $S\subseteq X$ is defined $\varphi(S):=\{\varphi(x):x\in S\}$,
and the preimage of $y\subseteq Y$ is defined $\varphi^{-1}(y):=\{x\in X:\varphi(x)=y\}$.
Let $(V, \dist)$ be an underlying metric space.
In \Cref{sec:local_frt} we use this general metric space notation,
while in later sections we assume it is Euclidean $\mathbb{R}^d$, i.e., $V = \mathbb{R}^d$ and $\dist = \ell_2$.
We use $\dist_{\mathcal{T}}(x,y)$ to denote the distance between points $x$ and $y$ in a graph $\mathcal{T}$,
which is supposed to be the embedding tree.
Let $B(p,r):=\{q:\dist(p,q)\leq r\}$ be the metric ball.
For a point set $S$, let $\diam(S) := \max_{x, y \in S} \dist(x, y)$ be the diameter.

\paragraph{Tree Embedding.} A \emph{tree embedding} for point set $P \subseteq V$ is a weighted tree $\mathcal{T}$, called the \emph{embedding tree}, together with a map $f$ from $P$ to the leaves of $\mathcal{T}$,
such that
a) map $f$ is a bijection, i.e. the leaf nodes in $\mathcal{T}$ have a one-to-one correspondence to $P$  and the other nodes are internal nodes, and
b) the length of the tree edges in the same level is the same, and in the leaf-to-root path the edge length doubles every hop.
Such tree is also known as a $2$-HST (hierarchically separated tree).
We say a (randomized) $\mathcal{T}$ under map $f$ has distortion $\beta \geq 1$ if for every $x, y \in P$, we have $\dist_{\mathcal{T}}(f(x), f(y)) \geq \dist(x, y)$ and $\E[\dist_{\mathcal T}(f(x), f(y))] \leq  \beta \cdot \dist(x, y)$. With a slight abuse of notation, we write $\dist_{\mathcal T}(x, y)$ instead of $\dist_{\mathcal T}(f(x), f(y))$ when $f$ is clear from context.

\paragraph{Johnson Lindenstrauss Transform.} For a set of $n$ points $P$ in $\RR^d$ the Johnson Lindenstrauss transform is a mapping $f: \RR^d \rightarrow \RR^k$ which satisfies that $ \|x-y\|_2 = \Theta(\|f(x) - f(y)\|_2)$. Crucially, \cite{JL84} has shown that such a transformation exists and is efficiently computable for any $k = \Theta(\log n)$. This essentially allows us to assume that $d = O(\log n)$ at the loss of a constant factor in distortion.

\paragraph{Assumptions and Parameters.}
For the sake of presentation, we assume without loss of generality that
the dataset $P \subseteq V$ has the smallest inter-point distance greater than $1$,
and its diameter, denoted as $\Delta := \diam(P)$, is $\Delta=2^{m-1} $ for some $m\in\mathbb{N}$.
For integer $i$, let $w_i := 2^{m - i}$.
Then the embedding tree $\mathcal{T}$ has height $m$, 
and the edge weight from every level-$i$ node to every level-$(i+1)$ node is $w_i$.
Here, level $1$ is the root (which it consists of a singleton partition),
and level $m$ is the leaves which each consists of (at most) a single data point.

\section{Tree Embedding via Metric Decomposition}
\label{sec:local_frt}

In this section we consider the general metric space setting and present our framework for tree embedding, listed in \Cref{alg:frt}.
Recall that $(V, \dist)$ is the notation for the underlying metric space.
This algorithm receives as input a dataset $P \subseteq V$, a parameter $\Gamma \geq 1$,
and assumes black-box access to a series of $m$ generic metric hash $\varphi_1, \ldots, \varphi_m$,
such that for each $i \in [m]$, $\varphi_i : V \to U$ for some domain $U$
satisfies that for each image $z \in \varphi(V)$, the diameter of the preimage (which may be viewed as a hash bucket which we define as a notation later) is bounded by $\tau_i$,
i.e., $\diam(\varphi^{-1}(z)) \leq \tau_i$ where $\tau_i := w_i / 2$.
The technical core of this algorithm is a modified CKR decomposition~\cite{CalinescuKR04} operated on general metric hashing,
where the major new step is Line~\ref{line:kpi}.
As we mentioned in~\Cref{sec:tech_overview}, this new Line~\ref{line:kpi} introduces nontrivial technical challenges
which requires new ideas in the analysis.

This algorithm serves as an outline for our dynamic and MPC algorithms,
and in this section, we only focus on the distortion analysis, i.e., \Cref{lemma:local_frt_distortion},
since the efficiency (in terms of time/round complexity etc.) depends on the actual dynamic/MPC implementations which will be discussed in later sections.
We denote by $\mathcal{T}$ the (implicit) embedding tree generated by our algorithm.

\paragraph{Implicit Embedding Tree $\mathcal{T}$.}
Instead of explicitly building an embedding tree $\mathcal{T}$,
the output of the algorithm is a sequence of labels 
$(\ell_p^{(1)}, \ldots, \ell_p^{(m)})$ for each $p \in P$.
This is slightly different to classical constructions~\cite{DBLP:journals/ejcss/FakcharoenpholRT04},
where the label is built in a top-down recursive manner (whereas our labels are built independently at every level).
Our version is similar to that used in previous parallel algorithms for tree embedding~\cite{GhaffariL14}.
To build the tree $\mathcal{T}$ implicitly, for $p \in P$, its labels $(\ell^{(1)}_p, \ldots, \ell^{(m)}_p)$ identify
each tree nodes from leaf $p$ to the root, as follows:
for $i \in [m]$, the $i$-prefix $(\ell^{(1)}_p, \ldots, \ell^{(i)}_p)$ constitutes the ID of the level-$i$ ancestor of leaf $p$.
The edge length from a level-$i$ node to level-$(i + 1)$ node is $w_i$, as we discussed in \Cref{sec:prelim}.

\paragraph{Notations for \Cref{alg:frt}.}
For each $i\in[m], x\in V$, define $\buk_i(x):=\varphi_i^{-1}(\varphi_i(x))$, which contains the points that are mapped to the same bucket as $x$ by $\varphi_i$
For each $i\in[m], S\subseteq V$, define $\buckets_i(S):=\bigcup_{x\in S} \buk_i(x)$, the union of all buckets that intersect with $S$, and define $\buckets_i^P(S):=\buckets_i(S)\cap P$.
We interpret a point as a singleton set, and because of this,
in the remainder of the paper we only use the notation $\buckets$ (and $\buk$ is not needed anymore).
For each $p\in P,r>0,i\in[m]$, define $\tilde B_i^P(p,r):=\buckets_i^P(B(p,r))$, which is the union of all data points in the buckets that intersect with $B(p,r)$.
For a set $S \subseteq P$, write $\pi_{\min}(S) := \min\{ \pi(u) : u \in S \}$.

\begin{algorithm}[H]
\DontPrintSemicolon
\caption{Tree embedding on input $P \subseteq V$ with access to metric hashes $\{ \varphi_i \}_{i \in [m]}$}
\label{alg:frt}

\tcc{$\varphi_i$ is metric hashing with diameter bound $\tau_i := w_i/2$}
Sample $\beta \in [\frac{1}{4}, \frac{1}{2}]$ uniformly at random\;

Let $\pi$ be a uniform random map from $P$ to $[0,1]$\;

\For{$i \gets 1$ \KwTo $m$}{
    $r_i \gets \beta / \Gamma \cdot w_i$\;
    
    for each $p \in P$,
        $\ell_p^{(i)} \gets \pi_{\min}(\tilde B_i^P(p, r_i))$ \label{line:kpi}\;
        
        \tcc{recall $\tilde B_i^P(p,r) = \buckets_i^P(B(p,r))$, and $\pi_{\min}(S) = \min\{ \pi(u) : u \in S \}$}
    
}
\Return{$\ell_p^{(i)}$ for $i \in [m]$, $p \in P$}\;
\end{algorithm}

The main result of this section establishes that \Cref{alg:frt} computes a tree embedding achieving near-optimal distortion, as stated below.

\begin{lemma}
    \label{lemma:local_frt_distortion}
    Consider an underlying metric $(V, \dist)$.
    Suppose \Cref{alg:frt} takes a point set $P \subseteq V$ and parameter $\Gamma \geq 1$ as input,
    and is given access to a series of metric hash $\varphi_1, \ldots, \varphi_m$ each with diameter bound $\tau_i := w_i / 2$,
    Let $\mathcal{T}$ be the corresponding embedding tree with respect to the output labels of \Cref{alg:frt},
    then $\forall p, q \in P$, it holds that $\dist_{\mathcal{T}}(p, q) \geq \dist(p, q)$ and
    $\E[\dist_{\mathcal{T}}(p, q)] \leq O(\Gamma \log \Gamma) \cdot \log n \cdot \dist(x, y)$.
\end{lemma}
\begin{proof}
    The proof can be found in \Cref{sec:lemma_local_frt_distortion}.
\end{proof}

\subsection{Proof of \Cref{lemma:local_frt_distortion}: Distortion Analysis of \Cref{alg:frt}}
\label{sec:lemma_local_frt_distortion}

In this section, we prove that the embedding tree corresponding to the output of \Cref{alg:frt} achieves the desired distortion guarantees. To this end, we begin by defining the distance between two points $p, q \in P$ in $\mathcal{T}$ in terms of their corresponding labels.

\paragraph{Tree distance with respect to labels.}

It is a standard fact that the distance between points $p, q$ is in the tree $\mathcal{T}$
is determined by their LCA (least common ancestor) in the tree.
The following notation defines the level of the LCA (minus $1$) with respect to the label $\ell$'s.
\begin{definition}
\label{def:level}
    For $p, q \in P$, define $\level(p, q)$ be the smallest $i \in [m]$ such that
    $\ell_p^{(i)} \neq \ell_q^{(i)}$.
\end{definition}
Since the embedding tree is $2$-HST, it is immediate that
\begin{equation}
\label{eqn:distT}
    \dist_{\tree}(p,q)=\sum_{i=\level(p,q)}^{m} 2 w_i
    \in [2w_{\level(p,q)},4w_{\level(p,q)}).
\end{equation}

Our goal is to prove that the \emph{dominating property} $\dist_T(p, q) \geq \dist(p, q)$ and the \emph{expected expansion} $\E[\dist_{\mathcal{T}}(p, q)] \leq O(\Gamma \log \Gamma) \cdot \log n \cdot \dist(x, y)$ hold for all $p, q\in P$.

\paragraph{Part I: Dominating Property $\dist_T(p, q) \geq \dist(p, q)$.}
    Fix $p,q\in P$, let $i:=\level(p,q)-1$.
    By \eqref{eqn:distT} we can see that $\text{dist}_T(p,q)\geq 2w_{i+1}$, 
    and by the definition of $\level(p,q)$ we have $\ell_p^{(i)}=\ell_q^{(i)}$.
    Let $p':=\pi^{-1}(\ell_p^{(i)})$, then we also have $p'=\pi^{-1}(\ell_q^{(i)})$.
    
    We claim that $\buckets_i(p') \cap B(p,r_i)\neq \emptyset$ and $\buckets_i(p') \cap B(q,r_i)\neq\emptyset$.
    To see this,
    observe that $p'\in \tilde B_i^P(p,r_i)$ (by the definition of $\ell_p^{(i)}$ as in Line~\ref{line:kpi}).
    Then by the definition of $\tilde B_i^P$, we can see $\varphi_i(p')\in \varphi_i(B(p,r))$, so $\buckets_i(p')\cap B(p,r_i)\neq \emptyset$.
    The argument regarding $\buckets_i(p') \cap B(q, r_i)$ is similar.

    Now, let $\hat p$ be an arbitrary point in $\buckets_i(p') \cap B(p,r_i)$, $\hat q$ be an arbitrary point in $\buckets_i(p') \cap B(q,r_i)$.
    Since $\hat p$ and $\hat q$ belong to the same bucket $\buckets_i(p')$, i.e., $\varphi_i(\hat p) = \varphi_i(\hat q)$,
    by the diameter bound of $\varphi$ (see \Cref{def:consistent_hashing}),
    we have $\dist(\hat p,\hat q)\leq \tau_i=w_{i+1}$.
    Also observe that $\dist(p,\hat p)\leq r_i$, $\dist(q,\hat q)\leq r_i$. Hence,
    \begin{align*}
        \dist(p,q)&\leq\dist(p,\hat p)+\dist(\hat p,\hat q)+\dist(\hat q,q)\\
        &\leq 2r_i+w_{i+1}\\
        &\leq 2w_{i+1}\\
        &\leq \dist_T(p,q).
    \end{align*}
    This concludes the dominating property.

The remainder of this section is devoted to the more challenging task of bounding the expected expansion of the tree embedding.

\paragraph{Part II: Expected Expansion $\E[\dist_T(p, q)] \leq O(\Gamma \log \Gamma) \cdot \log n \cdot \dist(p, q)$.}
We start by fixing $p,q\in P$.
By \eqref{eqn:distT}, we have
\begin{equation}
\label{eqn:exp_to_sep}
    \E[\dist_T(p,q)]
    \leq \sum_{i=1}^{m}\Pr[\level(p,q)=i]\cdot 4w_i
\end{equation}
Next, define $i_{\dist} :=m-\lceil \log_2 \dist(p,q)\rceil,i_{\max}:=i_{\dist}+3$, $i_{\min}:=i_{\dist}-\lceil\log_2(8\Gamma^2)\rceil$.
    We divide the levels $i$ into three intervals with respect to how large the scale $w_i$ compared with $\dist(p, q)$:

\begin{enumerate}
     \item \label{case:large_case} \textbf{``Large-distance'' case $i\in[1,i_{\min}]$.} This case is the most challenging one, and here we will show that $$\sum_{i=1}^{i_{min}}\Pr[\level(p,q)=i]\cdot 4w_i = O(\Gamma\log\Gamma)\cdot \log n \cdot \dist(p,q).$$
        This is achieved by applying the key lemma \Cref{claim:bound_single_pair}. \item \label{case:similar_case} \textbf{``Similar-distance'' case $i\in[i_{\min}+1,i_{\max}]$.} This case  is an easy case, since the number of levels $i$ that belong to this case is small, i.e., $O(\log\Gamma)$. We will show that $$\sum_{i=i_{min}+1}^{i_{max}}\Pr[\level(p,q)=i]\cdot 4w_i=O(\Gamma\log n)\cdot \dist(p,q),$$
        using a weaker bound \Cref{corollary:single_pair}.
     \item \label{case:small_case} \textbf{``Small-distance'' case $i\in[i_{\max}+1,m]$.} This case is the easiest, and here we will show that $ \Pr[\level(p,q)=i]=0$, and hence $$\sum_{i=i_{max}+1}^{m}\Pr[\level(p,q)=i]=0.$$ \end{enumerate}

     We note that the proofs of \Cref{claim:bound_single_pair} and its consequence, \Cref{corollary:single_pair}, are deferred to \Cref{sec:proof_claim_bound_single_pair}.  
We now proceed to discuss these three cases in more detail.    

    \paragraph{Analysis of Small-distance case (\Cref{case:small_case}): $i\in[i_{\max}+1,m]$.}
    We start with the case of $i > i_{\max}$, and this is an easy case since it cannot actually happen.
    Specifically, for each $i\in[m]$, if $i=\level(p,q)$,
        then we have $w_{1+i_{\dist}}\leq \dist(p,q)\leq \dist_T(p,q)< 4w_i=w_{i-2}$, which implies $i\leq i_{\dist}+3 = i_{\max}$.
    Hence, 
    \begin{equation*}
        \forall i > i_{\max}, \qquad \Pr[\level(p,q)=i]=0,
    \end{equation*}

    and hence 
    \begin{equation}\label{eqn:geq-imax}        \sum_{i=i_{max}+1}^{m}\Pr[\level(p,q)=i]=0.
    \end{equation}

Now, we focus on the remaining two cases. Before we proceed, we need the notation $r_i^{\max}:=w_i/(2\Gamma)$, which is the maximum radius of the ball at some level $i$.
    Then $\frac12r_i^{\max}\leq r_i\leq r_i^{\max}$.

    \paragraph{Analysis of Large-distance case (\Cref{case:large_case}): $i\in[1,i_{\min}]$ .}
    We first show that $p$ and $q$ must be ``close'' in this case.
    Specifically, we argue that $B(p,r_{i+2}^{\max}) \subseteq B(q,\frac12r_i^{\max})$ (here since $p$ and $q$ are symmetric so it also holds after swapping $p, q$).
    To see this, since
    $r_i^{\max}
    =w_i/(2\Gamma)
    \geq w_{i_{\min}}/(2\Gamma)
    \geq w_{i_{\dist}}\cdot 4\Gamma
    \geq \dist(p,q) \cdot 4\Gamma$,
    then for each $x\in B(p,r_{i+2}^{\max})$ we have
    $$
    \dist(x,q)
    \leq \dist(x, p) + \dist(p, q)
    \leq r_{i+2}^{\max} + \dist(p,q)
    \leq r_i^{\max} \left(\frac1{4\Gamma}+\frac14\right)\leq \frac12r_i^{\max}
    $$
which implies that
    \begin{equation}
    \label{eqn:ball_contain}
        B(p,r_{i+2}^{\max}) \subseteq B(q,\frac12r_i^{\max}).
    \end{equation}
    
    Next, we give an upper bound for $\Pr[\level(p, q) = i]$.
    Observe that $\Pr[\level(p, q) = i] \leq \Pr[\ell_p^{(i)} \neq \ell_q^{(i)}]$.
    In the following lemma we give an upper bound for $\Pr[\ell_p^{(i)} \neq \ell_q^{(i)}]$. 
\begin{restatable}{lemma}{boundsinglepair}
    \label{claim:bound_single_pair}
    For every $p,q\in P$, $i\in[m]$ and $0 \leq r' \leq \frac12 r_i^{\max}$ such that
    $B(p, r') \subseteq B(q, \frac 1 2 r_i^{\max})$,
it holds that
    $$
    \Pr[\ell_p^{(i)} \neq \ell_q^{(i)}]\leq \frac{8\cdot\dist(p,q)}{r_i^{\max}}\left(H_{|\tilde B_i^P(p,r_i^{\max})|}-H_{|\tilde B_i^P(p,r')|}\right).
    $$
\end{restatable}
\begin{proof}
    The proof can be found in \Cref{sec:proof_claim_bound_single_pair}.
\end{proof}
Applying \Cref{claim:bound_single_pair} with $r' = r_{i + 2}^{\max}$ (which satisfies the condition because of \eqref{eqn:ball_contain}), we obtain
\begin{equation}
\label{eqn:first_HH}
    \Pr[\level(p, q) = i]
    \leq \Pr[\ell_p^{(i)} \neq \ell_q^{(i)}]
    \leq \frac{8\cdot\dist(p,q)}{r_i^{\max}}\left( H_{|\tilde B_i^P(p,r_i^{\max})|}-H_{|\tilde B_i^P(p,r_{i+2}^{\max})|}\right).
\end{equation}
Recall that our goal in this case is to bound $\sum_{i=i_{min}}^{m}\Pr[\level(p,q)=i]\cdot 4w_i$. To achieve this, we further bound \eqref{eqn:first_HH} so that terms in this sum can cancel, for which we require the following lemma.
\begin{lemma}
    \label{fact:subset_without_nested}
    For each $p\in P,i,j\in[m]$, let $t:=\lceil \log_2\Gamma\rceil+2$, then we have
    $$
    \tilde B_{i+t}^P(p,r_{i+t}^{\max})\subseteq B(p,r_i^{\max})\subseteq \tilde B_j^P(p,r_i^{\max}).
    $$
\end{lemma}
\begin{proof}
    $B(p,r_i^{\max})\subseteq \tilde B_j^P(p,r_i^{\max})$ follows from the definition of $\tilde B_j^P$.
    It remains to show $\tilde B_{i+t}^P(p,r_{i+t}^{\max})\subseteq B(p,r_i^{\max})$.
Consider an arbitrary $q \in \tilde B_{i+t}^P(p,r_{i+t}^{\max})$ and we need to show $q \in B(p, r_i^{\max})$.
    Recall that $\tilde B_{i + t}^P(p, r_{i + t}^{\max}) = \buckets^P_{i + t}(B(p, r_{i + t}^{\max}))$,
    then there exists some $p'\in B(p,r_{i+t}^{\max})$ such that $\varphi_{i+t}(p')=\varphi_{i+t}(q)$.
    By the diameter guarantee of $\varphi$ (see \Cref{def:consistent_hashing}),
    we have $\dist(p,q)\leq \dist(p,p')+\dist(p',q)\leq r_{i+t}^{\max}+\tau_{i+t}=(1+2\Gamma)r_{i+t}^{\max}\leq 4\Gamma\cdot r_{i+t}^{\max}\leq r_{i}^{\max}$.
    This finishes the proof of \Cref{fact:subset_without_nested}.
\end{proof}
    Apply \Cref{fact:subset_without_nested} on \eqref{eqn:first_HH}
    and let $t := \lceil \log_2 \Gamma \rceil + 2$ be as in the lemma,
    \begin{align}
        \Pr[\level(p, q) = i]
        \leq \Pr[\ell_p^{(i)} \neq \ell_q^{(i)}]
        &\leq \frac{8\cdot\dist(p,q)}{r_i^{\max}}\left( H_{|\tilde B_i^P(p,r_i^{\max})|}-H_{|\tilde B_i^P(p,r_{i+2}^{\max})|}\right) \nonumber \\
        &\leq\frac{8\cdot\dist(p,q)}{r_i^{\max}}\left( H_{|\tilde B_i^P(p,r_i^{\max})|}-H_{|\tilde B_{i+t+2}^P(p,r_{i+t+2}^{\max})|}\right).
    \label{eqn:HH_second}
    \end{align}
With \eqref{eqn:HH_second}, we are now ready to conclude 
    \begin{equation}
        \label{eqn:leq-imin}
        \begin{aligned}
            \sum_{i=1}^{i_{\min}}\Pr[\level(p,q)=i]\cdot 4w_i
            &=\sum_{i=1}^{i_{\min}}\frac{8\cdot\dist(p,q)}{r_i^{\max}}\left( H_{|\tilde B_i^P(p,r_i^{\max})|}-H_{|\tilde B_{i+t+2}^P(p,r_{i+t+2}^{\max})|}\right) 4w_i\\
            &\leq 32\dist(p,q)\cdot\sum_{i=1}^{i_{\min}}\frac{d_i}{r_i^{\max}}\left( H_{|\tilde B_i^P(p,r_i^{\max})|}-H_{|\tilde B_{i+t+2}^P(p,r_{i+t+2}^{\max})|}\right)\\
            &\leq 128\Gamma\dist(p,q)\cdot\sum_{i=1}^{i_{\min}}\left( H_{|\tilde B_i^P(p,r_i^{\max})|}-H_{|\tilde B_{i+t+2}^P(p,r_{i+t+2}^{\max})|}\right)\\
            &\leq 128(t+2)\Gamma H_n\dist(p,q)\\
            &=O(\Gamma\log\Gamma)\cdot \log n\cdot \dist(p,q),
        \end{aligned}
    \end{equation}

    concluding the large-distance case (\Cref{case:large_case}).

    \paragraph{Analysis of Similar-distance case (\Cref{case:similar_case}): $i\in[i_{\min}+1,i_{\max}]$.}
Following a similar overall strategy with the previous case,
    we need the following lemma as a universal upper bound of the event $\ell_p^{(i)} \neq \ell_q^{(i)}$, without any other assumption on $p$ and $q$ (e.g., $B(p, r') \subseteq B(q, \frac12 r_i^{\max})$) as required in \Cref{claim:bound_single_pair}.
    \begin{restatable}{lemma}{corsinglepair}
        \label{corollary:single_pair}
        For each $p,q\in P$, $i\in[m]$, we have
        $$
        \Pr[\ell_p^{(i)} \neq \ell_q^{(i)}]\leq\frac{8H_n\cdot \dist(p,q)}{r_i^{\max}}.
        $$
    \end{restatable}
    \begin{proof}
        The proof can be found in \Cref{sec:proof_claim_bound_single_pair}.
    \end{proof}
    We conclude the similar-distane case by applying \cref{corollary:single_pair}
    \begin{align}
        \sum_{i=i_{\min} + 1}^{i_{\max}} \Pr[\level(p, q) = i] \cdot 4 w_i
        &\leq (i_{\max} - i_{\min}) \cdot \frac{16w_i}{r_i^{\max}}\cdot H_n\cdot \dist(p,q) \nonumber \\
        &\leq 64\Gamma \log\Gamma \cdot  H_n\cdot \dist(p,q) \nonumber \\
        &=O(\Gamma\log \Gamma \log n)\cdot \dist(p,q).
        \label{eqn:between-imin-imax}
    \end{align}

Finally, we are ready to prove \Cref{lemma:local_frt_distortion}

\paragraph{Concluding \Cref{lemma:local_frt_distortion}.}
    Summing the bounds obtained for the large-distance case \eqref{eqn:leq-imin}, the similar-distance case \eqref{eqn:between-imin-imax}, and the small-distance case \eqref{eqn:geq-imax} in equation \eqref{eqn:exp_to_sep}, we have
\begin{align*}
        \E[\dist_T(p,q)]
        &=\sum_{i=1}^{m}\Pr[\level(p,q)=i]\cdot 4w_i\\
        &=\sum_{i=1}^{i_{\min}}\Pr[\level(p,q)=i]\cdot 4w_i
            + \sum_{i=i_{\min}+1}^{i_{\max}}\Pr[\level(p,q)=i]\cdot 4w_i\\
        &\leq O(\Gamma\log\Gamma\log n)\cdot \dist(p,q),
    \end{align*}
    bounding the expected expansion of the tree embedding.
    This finishes the proof of \Cref{lemma:local_frt_distortion}. \qed

\subsection{Proof of \cref{claim:bound_single_pair} and \Cref{corollary:single_pair}}
\label{sec:proof_claim_bound_single_pair}

In this section, we focus on proving the two key lemmas, \Cref{claim:bound_single_pair} and \Cref{corollary:single_pair}, which are used  in the proof of \Cref{lemma:local_frt_distortion} in the analysis of the large-distance (\Cref{case:large_case}) and similar-distance (\Cref{case:similar_case}) cases, respectively. The proof of the latter lemma, \Cref{corollary:single_pair}, follows closely from that of \Cref{claim:bound_single_pair} and is briefly outlined at the end.

\boundsinglepair*

Although a similar statement is used in previous analysis of tree embedding algorithms such as~\cite{DBLP:journals/ejcss/FakcharoenpholRT04},
our proof deviates significantly.
A key challenge is that the label $\ell_p^{(i)}$ is defined with respect to
a geometric object $\tilde B_i^P(p, r_i)$,
which is more complicated than a standard ball $B(p, r_i)$ on which the argument of previous analysis builds.
To see this geometric complication, consider $p, q$ such that $\dist(p, q) = r$
then both $p \in \tilde B_i^P(q, r)$ and $p \not\in \tilde B_i^P(q, r)$ can possibly happen depending on $\tilde B_i^P$,
whereas for metric balls clearly $p \in B(q, r)$.
In a sense, this breaks the symmetry (and triangle inequality) and in turns makes previous analysis inapplicable.

Before presenting the proof of \Cref{claim:bound_single_pair}, we introduce two key concepts that help us overcome these difficulties, namely the representative set and the key representative.

\paragraph{Representative Set.}
Our proof strategy is to utilize the structure of $\tilde B_i^P(\cdot, \cdot)$,
and argue that a carefully chosen set of \emph{representatives} of $\tilde B_i^P(p, r)$ has similar geometry properties as a standard ball $B(p, r)$.
Specifically, since $\tilde B_i^P$ is a union of buckets,
our plan is to pick from each bucket exactly one point define a representative set for $\tilde B_i^P(p, r)$.
This representative set aligns better with the geometry of $B(p, r)$,
and we formally define the following properties that a proper representative set should satisfy.
We show in \Cref{lemma:rep_exist} that such proper representative set always exists.
Then in the remainder of the proof, we work with this representative set instead of directly with $\tilde B_i^P$.
While such representative set reveals the major structure of $\tilde B_i^P$,
we note that it is only used in the analysis and the algorithm does not rely on it.
\begin{definition}[Representative sets.]
\label{def:rep}
    A collection of sets $\rep_{i}(p, r) \subseteq V$ (which may not be a subset of $P$) defined for every $p \in P, i \in [m]$ and $r \geq 0$,
    are representative sets,
    if for every $i \in [m], p \in P, r \geq 0$, 
    \begin{enumerate}
        \item \label{item:rep_distinct}
        (distinct)
        $\forall x \neq y \in \rep_i(p, r)$, $\buckets_i(x) \cap \buckets_i(y) = \emptyset$;
\item
        \label{item:rep_monotone}
        (monotone) $\forall r' \in (0, r)$,
        $\rep_i(p, r') \subseteq \rep_i(p, r)$;
        \item
        \label{item:rep_ball}
        (ball-preserving)
        $\rep_i(p, r) \subseteq B(p, r)\cap\buckets_i(P)$ and
        $\buckets_i^P(\rep_i(p, r)) = \buckets_i^P(B(p, r))$.
\end{enumerate}
\end{definition}

Now we turn to showing that such representative sets exist. 
A natural attempt to define $\rep(\cdot, \cdot)$ such that it satisfies all three of distinct, monotone and ball-preserving properties,
is to let $\rep_i(r, p) := B(p, r) \cap \buckets_i(P)$, which is precisely the first half of the property \Cref{item:rep_ball}.
The issue is that it does not satisfy the distinctness \Cref{item:rep_distinct}.
Our definition presented in the \Cref{lemma:rep_exist} below fixes this by only selecting a subset of $B(p, r) \cap \buckets_i(P)$,
roughly the nearest neighbors from $p$ in each bucket of $B(p, r) \cap \buckets_i(P)$.
This ``nearest'' property ensures the distinctness (\Cref{item:rep_distinct}).
Using the geometric properties of the nearest neighbor, we can also show that such sets satisfy also properties \Cref{item:rep_monotone} and \Cref{item:rep_ball}.

\begin{restatable}{lemma}{lemmarepexist}
    \label{lemma:rep_exist}
    There exists $\{\rep_i(p, r) : i \in [m], p \in P, r \geq 0\}$ that satisfies \Cref{def:rep}.
\end{restatable}
\begin{proof}
We need the following notations to define $\rep$.
For each $i\in[m],p\in P,x\in V$, define
\begin{equation}
    \NN^{\buckets}_i(p, x):=\arg\min_{y\in \buckets_i(x)}\dist(p, y)\in V,\footnote{For each $i\in [m],p\in P,x\in V$, we assume without loss of generality that the nearest point in $\buckets_i(x)$ to $p$ is unique.
}
\end{equation}
the nearest point in $\buckets_i(x)$ to $p$, which may not be in $P$.
Then for each $i\in[m],p\in P,S\subseteq V$, define $\NN^{\buckets}_i(p,S):=\{\NN^{\buckets}_i(p,x):x\in S\}$.
Note that for each $x,y\in S$, if $\varphi_i(x)=\varphi_i(y)$, then $\NN^{\buckets}_i(p,x)=\NN^{\buckets}_i(p,y)$,
    so $\NN^{\buckets}_i(p,S)$ contains at most one point from each bucket.

\paragraph{Defining $\rep$.}
For each $p\in P, i\in [m], r>0$, we define 
\begin{equation}
\label{eqn:rep}
    \rep_i(p,r):=\NN^{\buckets}_i(p, B(p,r)\cap\buckets_i(P)).
\end{equation}
That is, for each bucket that intersects with $B(p,r)$, we choose the nearest point to $p$ as the representative point.
It remains to verify that this definition of $\rep$ satisfies \Cref{def:rep}.
Fix some $p\in P,i\in[m],r\geq 0$.

\paragraph{Verifying \Cref{item:rep_distinct}.}
Recall \Cref{item:rep_distinct} states that for each $x\neq y\in\rep_i(p,r)$, $\buckets_i(x)\cap\buckets_i(y) = \emptyset$.
Consider some $x\neq y\in\rep_i(p,r)$.
Suppose for the contrary that $\buckets_i(x)\cap\buckets_i(y)\neq\emptyset$,
then there exists some $z\in\buckets_i(x)\cap\buckets_i(y)$,
which implies that $\varphi(x)=\varphi(z)=\varphi(y)$.
But by the definition of $\rep_i(p,r)$, at most one point of each bucket may belong to $\rep_i(p,r)$ which is a contradiction.
This verifies \Cref{item:rep_distinct}.

\paragraph{Verifying \Cref{item:rep_monotone}.}
Recall \Cref{item:rep_monotone} states that $\forall r'\in(0,r)$, $\rep_i(p,r')\subseteq\rep_i(p,r)$.
Observe that for each $r'\in(0,r)$, we have $B(p,r')\cap\buckets_i(P)\subseteq B(p,r)\cap\buckets_i(P)$.
On the other hand, for sets $S, T \subseteq \mathbb{R}^d$ such that $S \subseteq T$, it is immediate that $\NN_i^{\buckets}(p,S)\subseteq\NN_i^{\buckets}(p,T)$.
Then the item follows by
substituting $S = B(p,r')\cap\buckets_i(P)$ and $T = B(p,r)\cap\buckets_i(P)$
into this fact.

\paragraph{Verifying \Cref{item:rep_ball}.}
Recall \Cref{item:rep_ball} states that $\rep_i(p, r) \subseteq B(p, r)\cap\buckets_i(P)$ and
$\buckets_i^P(\rep_i(p, r)) = \buckets_i^P(B(p, r))$.
We first prove that $\rep_i(p, r) \subseteq B(p, r)\cap\buckets_i(P)$.
By definition, for each $x\in\rep_i(p,r)$, there exists $y\in B(p,r)\cap\buckets_i(P)$ such that $x=\NN_i^{\buckets}(p,y)$.
Now, observe that $\dist(p,\NN_i^{\buckets}(p,y)) \leq \dist(p,y)$,
which implies $\dist(p,x)\leq\dist(p,y)\leq r$, that is, $x\in B(p,r)$.
On the other hand, using the definition of $\NN$ we have $x\in\buckets_i(y)$, which means $\varphi_i(x)=\varphi_i(y)$. Besides, using $y\in\buckets_i(P)$ there exists $z\in P$ such that $\varphi(x)=\varphi(y)=\varphi(z)$, which implies that $x\in\buckets_i(P)$. Thus, $x\in B(p,r)\cap\buckets_i(P)$.

Next we prove that $\buckets_i^P(B(p,r))=\buckets_i^P(\rep_i(p,r))$. Using $\rep_i(p,r)\subseteq B(p,r)$ we can see that $\buckets_i^P(\rep_i(p,r))\subseteq\buckets_i^P(B(p,r))$. Then we prove that $\buckets_i^P(B(p,r))\subseteq \buckets_i^P(\rep_i(p,r))$.
For each $x\in\buckets_i^P(B(p,r))$, there exists $x'\in B(p,r)$ that satisfies $\varphi_i(x)=\varphi_i(x')$. Due to $x\in P$ and $x'\in\buckets_i(x)$, we have $x'\in B(p,r)\cap \buckets_i(P)$. Let $y:=\NN_i^{\buckets}(p,x')$. We can see that $y\in\rep_i(p,r)$. Then using $x\in P$ and $x\in\buckets_i(y)$, we have $x\in\buckets_i^P(\rep_i(p,r))$.

This finishes the proof of \Cref{lemma:rep_exist}.
\end{proof}

\paragraph{Key Representative.}
In the remainder of the section, we assume $\{ \rep_i(p, r) \}_{i, p, r}$ are representative sets.
We then identify a key representative $\rep^*_i(p, r) \in \rep_i(p, r)$ (for $i, p, r$) in \Cref{def:rep_star}.
Roughly speaking, this $\rep^*_i(p, r) \in \rep_i(p, r)$ is a representative point, coming from the bucket that contains a minimizer of the $\pi$-value within $\tilde B_i^P(p, r)$.

\begin{definition}[Key representative]
\label{def:rep_star}
	For each $i\in[m],p\in P, r>0$, define $\rep_i^*(p,r)$ as the minimum point in $\rep_i(p,r)$ with respect to $\pi$, i.e.
	$$
		\rep_i^*(p,r):=\arg\min\{\pi_{\min}(\buckets_i^P(x)):x\in \rep_i(p,r) \}.
	$$ 
    Note that $\rep_i^*(p,r)$ is unique, since all values of $\pi$ are distinct with probability $1$.
\end{definition}
By \Cref{item:rep_ball} of \Cref{def:rep}, we already know that the label $\ell_p^{(i)}$, which takes the min $\pi$-value in $\tilde B^P_i(p, r_i)$,
can be equivalently evaluated within that of $\rep_i(p, r)$. This $\rep^*_i(p, r)$ in \Cref{def:rep} precisely realizes the minimizer of $\pi$-value within $\rep_i(p, r)$: 
\begin{equation}
	\label{eqn:rep_star}
	\forall i, p, \quad \ell_p^{(i)}
    = \pi_{\min}(\buckets_i^P(\rep^*_i(p, r_i))).
\end{equation}
This fact can be formally implied by combining \Cref{def:rep_star} with \Cref{item:rep_ball} of \Cref{def:rep}.
The technical reason to consider this $\rep^*$ which may not be in $P$,
instead of the minimizer in $\tilde B^P(p, r)$ which is a subset of $P$,
is that it enables us to utilize properties in \Cref{def:rep}.

We now turn to the proof of \Cref{claim:bound_single_pair}.

\paragraph{Proof of \Cref{claim:bound_single_pair}.}
First, notice that the event $\ell_p^{(i)} \neq \ell_q^{(i)}$ is equivalent to
either $\ell_p^{(i)} < \ell_q^{(i)}$ or
$\ell_p^{(i)} > \ell_q^{(i)}$,
and these two events are symmetric. Therefore, we only focus on bounding one of them $\Pr[\ell_p^{(i)} < \ell_q^{(i)}]$.

Specifically, we show that the event $\ell_p^{(i)} < \ell_q^{(i)}$, implies a geometric fact regarding the key representative $\rep^*$; specifically,
$\rep^*_i(p, r_i) \not\in B(q, r_i)$.
To see this, by \eqref{eqn:rep_star} and Line~\ref{line:kpi} of \Cref{alg:frt},
we have
$\ell_p^{(i)} = \pi_{\min}(\buckets_i^P(\rep^*_i(p, r_i))) < \ell_q^{(i)} = \pi_{\min}(\buckets^P_i(B(q, r_i))$, 
and this implies that $\buckets_i^P(\rep_i^*(p, r_i))$ cannot be a subset of $\buckets_i^P(B(q, r_i))$,
which leads to $\rep_i^*(p, r_i) \not\in B(q, r_i)$.
Therefore,
\begin{align}
\label{eqn:first_step}
	\Pr[\ell_p^{(i)} < \ell_q^{(i)}] \leq \Pr[ \rep^*_i(p, r_i) \not\in B(q, r_i) ].
\end{align}

\paragraph{Analyzing $\Pr[\rep_i^*(p,r_i)\notin B(q,r_i)]$ via Auxiliary Event.}
One difficulty of analyzing {$\Pr[  \rep^*_i(p, r_i)\not\in B(q, r_i) ]$}
is that it depends on both sources of randomness (i.e., $\pi$ and $r_i$).
To separate these two sources of randomness so that we can argue independently,
we define an auxiliary event $\calE$ whose randomness only comes from $\pi$.
This event is defined with respect to $i \in [m]$, $p \in P$, and $x \in V$.
It concerns whether or not $x$ has the minimum $\pi$-value in $\tilde B_i^P(p, r)$,
for $r = \dist(p, x)$.
The detailed definition goes as follows.
\begin{definition}[Auxiliary event]
For every $i \in [m], p \in P, x \in V$, define an event with respect to the randomness of $\pi$
\begin{equation}
\label{eqn:aux_event}
    \calE_p^{(i)}(x) := \{  
    \pi_{\min}(\buckets^P_i(x))
    =
    \pi_{\min}(\buckets^P_i(\rep_i(p, r)))\},
\end{equation}
where $r = \dist(p, x)$. 
\end{definition}
We relate $\Pr[\rep^*_i(p, r_i) \not \in B(q, r_i)]$, which is what we need to prove as in \eqref{eqn:first_step},
to the event $\calE_p^{(i)}$ in the following \Cref{lemma:split_probability}.

\begin{lemma}
    \label{lemma:split_probability}
    For each $p\in P,i\in[m],x\in \rep_i(p,r_i^{\max})$, we have
    $$
    \Pr[x=\rep_i^*(p,r_i)\land x\notin B(q,r_i)]\leq\Pr[\calE_p^{(i)}(x)]\cdot\Pr[x\in B(p,r_i)\land x\notin B(q,r_i)]
    $$
\end{lemma}

\begin{proof}
\begin{align*}
        \Pr[x=\rep_i^*(p,r_i)\land x\notin B(q,r_i)]
        &= \Pr[x=\rep_i^*(p,r_i)\land x\in B(p,r_i)\land x\notin B(q,r_i)]\\
        &= \Pr[x=\rep_i^*(p,r_i)\land \calE_p^{(i)}(x) \land x\in B(p,r_i)\land x\notin B(q,r_i)]\\
        &\leq \Pr[\calE_p^{(i)}(x)\land x\in B(p,r_i)\land x\notin B(q,r_i)]\\
        &=\Pr[\calE_p^{(i)}(x)]\cdot \Pr[x\in B(p,r_i)\land x\notin B(q,r_i)]
    \end{align*}
    where the first equality follows from the fact that
    $x=\rep_i^*(p,r_i) \in \rep_i(p, r_i) \subseteq B(p, r_i)$ (by \Cref{item:rep_ball} of \Cref{def:rep}),
    the second equality uses that $\calE_p^{(i)}(\rep_i^*(p, r_i))$ always happens (which is immediate from \Cref{def:rep_star} and \Cref{eqn:aux_event}),
    and the last equality holds because $\calE_p^i(x)$ is independent to the randomness of $r_i$.
\end{proof}

    Continuing from \eqref{eqn:first_step},
\begin{align}
    \Pr[\ell_p^{(i)}<\ell_q^{(i)}]
    &\leq \Pr[\rep_i^*(p,r_i)\notin B(q,r_i)] \nonumber \\
    &= \sum_{x\in \rep_i(p,r_i^{\max})}\Pr[x=\rep_i^*(p,r_i)\land x\notin B(q,r_i)] \nonumber \\
    &\leq \sum_{x\in \rep_i(p,r_i^{\max})}\Pr[\calE_p^{(i)}(x)]\cdot \Pr[x\in B(p,r_i)\land x\notin B(q,r_i)],
    \label{eqn:split_probability}
    \end{align}
    where the equality follows from the fact that
    $\rep_i^*(p,r_i) \in \rep_i(p, r_i) \subseteq \rep_i(p,r_i^{\max})$ (which uses \Cref{item:rep_monotone} of \Cref{def:rep})
    and that the events $x=\rep_i^*(p,r_i)$ for every $x\in \rep_i(p,r_i^{\max})$ are mutually exclusive,
    and the last inequality follows from \Cref{lemma:split_probability}.

    Next, we bound each of
    $\Pr[\calE_p^{(i)}(x)]$ and $\Pr[x\in B(p,r_i)\land x\notin B(q,r_i)]$ for $x \in \rep_i(p, r_i^{\max})$,
    which are the two multipliers in each summation term of \Cref{eqn:split_probability}.
    We need the following notations.
    Let $s:=|\rep_i(p,r_i^{\max})|$, and write 
    $\rep_i(p,r_i^{\max}) := \{p^{(i)}_1,\ldots,p^{(i)}_s\}$
    such that the points $\{p^{(i)}_j\}_j$
    are listed in the non-decreasing order of distance to $p$.
    For each $j \in[s]$, let $T_{j}:=\bigcup_{t=1}^{j}\buckets_i^P(p^{(i)}_t)$,
    $t_{j}:=|T_{j}|$ and $b_{j}:=|\buckets_i^P(p^{(i)}_{j})|$.
    In this context, we fix some $j \in [s]$ and consider $x = p_j^{(i)}$.

\paragraph{Bounding $\Pr[\calE_p^{(i)}(p_j^{(i)})]$ in \eqref{eqn:split_probability}.}
We now aim to show the following bound.
    \begin{claim}\label{claim:first_term}
    For all $j\in [s]$, we have that $\Pr[\calE_p^{(i)}(p_j^{(i)})] \leq 2(H_{t_j}-H_{t_{j-1}})$.
    \end{claim}
    \begin{proof}
    By \eqref{eqn:aux_event},
    $\calE_p^{(i)}(p_j^{(i)})$ happens if and only if $\pi_{\min}(\buckets_i^P(p_j^{(i)}))$
    is the minimum in $\{ \pi_{\min}(\buckets_i^P(p_t^{(i)})) \}_{t \leq j}$.
    To calculate this probability,
    we consider the reverse event that 
    none of the elements in $\buckets_i^P(p_j^{(i)})$
    is the minimum among $\{ \pi_{\min}(\buckets_i^P(p_t^{(i)})) \}_{t \leq j - 1}$.
    This probability is $\left(1-\frac{1}{1+t_{j-1}}\right)^{b_j}$ where we are using the independence of $\pi$ value.
    Hence,
    $$
    \Pr[\calE_p^{(i)}(p_j^{(i)})]= 1-\left(1-\frac{1}{1+t_{j-1}}\right)^{b_j}.
    $$

Using Bernoulli's inequality, we have
    $$
    \Pr[\calE_p^{(i)}(p_j^{(i)})]\leq \frac{b_j}{1+t_{j-1}}.
    $$
    We then bound this quantity by a case analysis.
    
    \subparagraph{Case a): $b_j\leq t_{j-1}$.}
    This case is straightforward,
    $$
    \Pr[\calE_p^{(i)}(p_j^{(i)})]\leq \sum_{k=1}^{b_j}\frac{1}{1+t_{j-1}} \leq \sum_{k=1}^{b_j} \frac{2}{k+t_{j-1}}\leq 2(H_{t_j}-H_{t_{j-1}}).
    $$
    
    \subparagraph{Case b): $b_j>t_{j-1}$.}
    In this case,
    $$
        H_{t_j}-H_{t_{j-1}}\geq \sum_{i=t_{j-1}+1}^{t_j}\frac1i\geq \frac{b_j}{t_{j-1}+b_j}> \frac12,
    $$
    where the second-last inequality uses that $t_{j-1}+b_j=t_j$,
    which follows from \Cref{item:rep_distinct} of \Cref{def:rep} (so sets in $\{\buckets_i^P(p_t^{(i)})\}_t$ are disjoint).
    Then $\Pr[\calE_p^{(i)}(p_j^{(i)})]\leq 1\leq 2(H_{t_j}-H_{t_{j-1}})$ trivially holds.

    In conclusion, we have shown 
    \begin{equation}
    \label{eqn:first_term}
        \forall j \in [s], \qquad
    \Pr[\calE_p^{(i)}(p_j^{(i)})]\leq 2(H_{t_j}-H_{t_{j-1}}).
    \end{equation}
    \end{proof}

\paragraph{Bounding $\Pr[p_j^{(i)}\in B(p,r_i)\land p_j^{(i)}\notin B(q,r_i)]$  in \eqref{eqn:split_probability}.}
    We show the following bound. 
    
    \begin{claim}\label{claim:second_term}
    For all $j\in [s]$, we have that $\Pr[p_j^{(i)}\in B(p,r_i)\land p_j^{(i)}\notin B(q,r_i)]\leq \frac{2}{r_i^{\max}}\dist(p,q)$.
    \end{claim}
    Observe that if $p_j^{(i)}\in B(p,r_i)$
    and $p_j^{(i)}\notin B(q,r_i)$,
    then $\dist(p_j^{(i)}, p) \leq r_i \leq \dist(p_j^{(i)},q)$.
    Hence,
    $$
    \Pr[p_j^{(i)}\in B(p,r_i)\land p_j^{(i)}\notin B(q,r_i)]\leq \Pr[\dist(p_j^{(i)},p)\leq r_i\leq\dist(p_j^{(i)},q)].
    $$
    Recall that $r_i$ is uniformly distributed in $[\frac12r_i^{\max}, r_i^{\max}]$.
    Hence, we conclude that
    \begin{equation}
    \label{eqn:second_term}
    \forall j \in [s], \quad \Pr[p_j^{(i)}\in B(p,r_i)\land p_j^{(i)}\notin B(q,r_i)]\leq\frac{\dist(p_j^{(i)},q)-\dist(p_j^{(i)},p)}{\frac12r_i^{\max}}\leq\frac{2}{r_i^{\max}}\dist(p,q).
    \end{equation}
\paragraph{Concluding \Cref{claim:bound_single_pair}.}
    Before combining \eqref{eqn:split_probability}
    with bounds from \eqref{eqn:first_term} and \eqref{eqn:second_term},
    we need an addition step, which asserts that $p_j^{(i)}$'s that are close enough to $p$ has no contribution to \eqref{eqn:split_probability}.
    This claim uses the premises
    that $B(p, r') \subseteq B(q, \frac12 r_i^{\max})$,
    and this is the only place where we use this premises.
    
    Let $s':=|\rep_i(p,r_i^{\max})\cap B(p,r')|$.
    We show that $\Pr[p_j^{(i)}\in B(p,r_i)\land p_j^{(i)}\notin B(q,r_i)]=0$ for each $j\leq s'$.
    This requires a key property: 
        \begin{claim*}
            $\rep_i(p, r_i^{\max}) \cap B(p, r') = \rep_i(p, r')$.
        \end{claim*}
    \begin{proof}[Proof of Claim]
        We start with $\rep_i(p, r') \subseteq \rep_i(p, r_i^{\max}) \cap B(p, r')$.
        We know $\rep_i(p, r') \subseteq B(p, r')\cap\buckets_i(P)$
        by \Cref{item:rep_ball} of \Cref{def:rep},
        and $\rep_i(p, r') \subseteq \rep_i(p, r_i^{\max})$ by \Cref{item:rep_monotone} of \Cref{def:rep} and $r' \leq r_i^{\max}$.
        These imply one direction: $\rep_i(p, r') \subseteq \rep_i(p, r_i^{\max}) \cap B(p, r')$.
For the other direction,
        by \Cref{item:rep_ball} of \Cref{def:rep},
        \begin{align*}
            \buckets_i^P(\rep_i(p, r_i^{\max}) \cap B(p, r'))
            \subseteq \buckets_i^P(B(p, r')) = \buckets_i^P(\rep_i(p, r')),
        \end{align*}
        so for every $x \in \rep_i(p, r_i^{\max}) \cap B(p, r')$ there exists some $y \in \rep_i(p, r') \subseteq \rep_i(p, r_i^{\max})$
        such that
        $\buckets_i^P(x) = \buckets_i^P(y)$.
        Since we already established that
        $\rep_i(p, r') \subseteq \rep_i(p, r_i^{\max}) \cap B(p, r')$,
        so we have $y \in \rep_i(p, r_i^{\max} \cap B(p, r'))$ as well.
        Therefore, it must be the case that $x = y$, since other,
        by \Cref{item:rep_distinct} of \Cref{def:rep},
        we would have $\buckets_i(x) \cap \buckets_i(y) = \emptyset$,
            $\buckets_i^P(x)\cap\buckets_i^P(y)=\emptyset$,
            and this is a contradiction.
            This concludes the other direction, and finishes the proof of the claim.
    \end{proof}
    This claim implies $s'= |\rep_i(p, r_i^{\max}) \cap B(p, r')| = |\rep_i(p,r')|$.
    By \Cref{item:rep_monotone} of \Cref{def:rep},
    we have $\{p_1^{(i)}, \ldots, p_{s'}^{(i)} \} = \rep_i(p, r')$ (recalling that the full set $\{p_1^{(i)}, \ldots, p_s^{(i)}\} = \rep_i(p, r_i^{\max})$).
By the premises and \Cref{item:rep_monotone} of \Cref{def:rep},
    we have $\rep_i(p, r') \subseteq B(p, r') \subseteq B(q, \frac12 r_i^{\max})$,
    and this implies
    $p_j^{(i)}\in B(p,\frac12r_i^{\max})$ for each $j\leq s'$.
    Then using $r_i\geq \frac12r_i^{\max}$, we have $p_j^{(i)}\in B(p,r_i)$, that is, 
    \begin{equation}
    \label{eqn:zero_prob}
        \forall j \leq s', \qquad \Pr[p_j^{(i)}\in B(p,r_i)\land p_j^{(i)}\notin B(q,r_i)]=0.
    \end{equation}
    
Finally, by combining the bounds obtained in \eqref{eqn:first_term}, \eqref{eqn:second_term}, and \eqref{eqn:zero_prob} within equation \eqref{eqn:split_probability}, we have
    \begin{align*}
    \Pr[\ell_p^{(i)}<\ell_q^{(i)}]
    &\leq \sum_{j=1}^{s}\frac{2}{r_i^{\max}}\dist(p,q)\cdot 2(H_{t_j}-H_{t_{j-1}})\\
    &= \sum_{j=s'+1}^{s}\frac{2}{r_i^{\max}}\dist(p,q)\cdot 2(H_{t_j}-H_{t_{j-1}})\\
    &= \frac{2\cdot\dist(p,q)}{r_i^{\max}}\left(\sum_{j=s'+1}^{s} 2(H_{t_j}-H_{t_{j-1}})\right)\\
    &= \frac{4\cdot\dist(p,q)}{r_i^{\max}}\left(H_{|\tilde B_i^P(p,r_i^{\max})|}-H_{|\tilde B_i^P(p,r')|}\right)\\
    \end{align*}
    This finishes the proof of \Cref{claim:bound_single_pair}. 
    \qed

We now turn to the proof of \Cref{corollary:single_pair}, used in the similar-distance analysis (\Cref{case:similar_case}) of \Cref{claim:bound_single_pair}.

\corsinglepair*

\begin{proof}[Proof of \Cref{corollary:single_pair}]
    Compared with \Cref{claim:bound_single_pair},
    the main difference in  \Cref{corollary:single_pair} 
    is that it is without the premises $B(p, r') \subseteq B(q, \frac12 r_i^{\max})$.
    Observe that it is only near the end of the proof of \Cref{claim:bound_single_pair}
    where we use this premises, which is precisely to show \eqref{eqn:zero_prob}.
Hence, \eqref{eqn:split_probability}, \eqref{eqn:first_term}, \eqref{eqn:second_term} are still valid for \Cref{corollary:single_pair},
    and we can obtain the following.
    \begin{align*}
        \Pr[\ell_p^{(i)}<\ell_q^{(i)}]
        &\leq \sum_{j=1}^{s}\frac{2}{r_i^{\max}}\dist(p,q)\cdot 2(H_{t_j}-H_{t_{j-1}})\\
        &\leq \frac{2\cdot\dist(p,q)}{r_i^{\max}}\cdot 2 H_n\\
        &\leq \frac{4H_n \cdot \dist(p, q)}{ r_i^{\max} }.
    \end{align*}
    Thus, we have 
    \begin{align*}
        \Pr[\ell_p^{(i)}\neq\ell_q^{(i)}]=\Pr[\ell_p^{(i)}<\ell_q^{(i)}]+\Pr[\ell_p^{(i)}>\ell_q^{(i)}]\leq \frac{8H_n \cdot \dist(p, q)}{ r_i^{\max} }.
    \end{align*}
    This finishes the proof of \Cref{corollary:single_pair}.
\end{proof}

\section{Dynamic and MPC Implementations of \Cref{alg:frt}}
\label{sec:mpc_dyna}

In this section we present how to implement \Cref{alg:frt} efficiently for Euclidean $\mathbb{R}^d$,
in dynamic and MPC settings.
The key is to implement Line~\ref{line:kpi}, particularly to bound the size of $\tilde B$'s.
To this end, we make use of the following notion of consistent hashing.

\begin{definition}[\cite{CJKVY22,filtser2025fasterapproximationalgorithmskcenter}]
    \label{def:consistent_hashing}
    A (randomized) hash function $\varphi:\RR^d\to\RR^d$ is a \emph{$\Gamma$-gap $\Lambda$-consistent hash} with diameter bound $\tau > 0$, or simply a $(\Gamma,\Lambda)$-hash function, if it satisfies:
    \begin{enumerate}
        \item (Diameter) for every image $z\in \varphi(\RR^d)$, we have $\diam(\varphi^{-1}(z))\leq\tau$.
        \item (Consistency) for every $S\subseteq \RR^d$ with $\diam(S)\leq\tau/\Gamma$, we have $\E[|\varphi(S)|] \leq \Lambda$.
    \end{enumerate}
\end{definition}
The notion of consistent hashing was introduced by~\cite{CJKVY22} which is deterministic,
and the notion in \Cref{def:consistent_hashing} is a relaxed randomized version which was introduced by~\cite{filtser2025fasterapproximationalgorithmskcenter}.
It turns out that the expected guarantee of consistency already suffices for our applications.
For example, if we have a $(\Gamma, \Lambda)$-hash, then $\E[|\tilde B^P_i(p, r_i)|] \leq \Lambda$ (as in Line~\ref{line:kpi}).
We state our theorems for dynamic and MPC tree embedding with respect to a general $(\Gamma, \Lambda)$-hash, as follows.
Their proof can be found in \Cref{sec:proof_thm_mpc,sec:proof_thm_dynamic}.

\begin{theorem}
    \label{thm:dynamic}

    Assume there exists a $(\Gamma,\Lambda)$-hash $\varphi:\RR^d \rightarrow \RR^d$ with diameter bound $\tau$ such that for $p \in \RR^d$ hash value $\varphi(p)$ and set of hash values $\varphi(B(p,\tau/\Gamma))$ can be evaluated in $O(\poly(d))$ and $O(|\varphi(B(p,\tau/\Gamma))| \cdot \poly(d))$ times respectively.

    Then there exists a dynamic algorithm which for dynamic set of points $P \subseteq \RR^d, |P| \leq n$ undergoing point insertions and deletions
maintains a tree-embedding of $P$ with $O(\Gamma \log \Gamma \log n)$ distortion in $\tilde{O}(d + \Lambda)$ expected amortized update time. 

    The underlying tree-embedding is rebuilt by the algorithm after every $n$ updates. An update to the input points $P$ results in $\tilde{O}(1)$ expected updates to the tree-embedding of the following types:

    \begin{itemize}
        \item Type 1): A leaf of the embedding becomes inactive in the sense that no point in $P$ corresponds to it.
        \item Type 2): A new leaf and a path connecting the leaf to an existing node in the tree embedding are inserted into the embedding.
    \end{itemize}
    
\end{theorem}

\begin{restatable}[MPC tree embdding]{theorem}{thmmpc}
    \label{thm:mpc}
    Assume there exists a family of $(\Gamma, \Lambda)$-hash $\varphi : \RR^d \to \RR^d$ such that it may be described in $\poly(d)$ space.
    There is an MPC algorithm that takes as input
    a set of $n$ points $P\subseteq [\Delta]^d$ distributed across machines with local memory $s = \Omega(\poly(d \log \Delta))$,
    computes a tree embedding with $O(\Gamma\log\Gamma \log n)$ distortion
    in $O(\log_s n)$ rounds and $O(n\poly(d\log \Delta) \cdot \Lambda)$ total space.
\end{restatable}

\paragraph{Consistent Hashing Bounds.}
We discuss the concrete hashing bounds that can be plugged into \Cref{thm:dynamic,thm:mpc}.
\cite{filtser2025fasterapproximationalgorithmskcenter} gave a hash function that satisfies the requirements of \cref{thm:dynamic}, for $\Lambda = 2^{O(d / \Gamma^{2 / 3})} \cdot \poly(d)$.
Moreover, in~\cite{CJKVY22} a hashing that satisfies the requirements of \Cref{thm:mpc} was presented and $\Lambda = 2^{O(d / \Gamma)} \cdot \poly(d)$.
These two results are (re)stated as follows.

\begin{lemma}[\cite{filtser2025fasterapproximationalgorithmskcenter}]
    \label{lemma:hashing_dynamic}
    For each $\Gamma \geq \sqrt{2\pi}$, there exists a $(\Gamma,\Lambda)$-hash $\varphi:\RR^d\to\RR^d$ for $\Lambda = \exp(O(d / \Gamma^{\frac23}))\poly(d)$ with diameter bound $\tau>0$ where $\varphi$ can be evaluated in $\poly(d)$ time for every input in $\mathbb{R}^d$.
    Furthermore, given a point $p \in \RR^d$, the set of hash values $\varphi(B(p, \tau/ \Gamma))$ can be evaluated in $O(|\varphi(B(p, \tau / \Gamma))| \cdot \poly(d))$ time.
\end{lemma}
The second property in \Cref{lemma:hashing_dynamic} is not explicitly stated, but can be concluded from \cite{filtser2025fasterapproximationalgorithmskcenter}. We do not claim novelty for it and elaborate on its execution in \Cref{app:hashing_localsearch}.

\begin{lemma}[{\cite[Theorem 5.1]{CJKVY22}}]
   \label{lemma:hashing}
   For each $\Gamma\in[8,2d]$, there exists a deterministic $(\Gamma,\Lambda)$-hash $\varphi:\RR^d\to\RR^d$ for $\Lambda=\exp(8d/\Gamma)\cdot O(d\log d)$ where $\varphi$ can be described by $O(d^2\log^2d)$ bits and can be computed in $\exp(d)$ time.
\end{lemma}

\paragraph{Dimension Reduction.}
Although these bounds have an exponential in $d$ dependence in $\Lambda$,
they are actually enough for implying the main theorems we state in \Cref{sec:results},
since one can without loss of generality applying a standard Johnson-Lindenstrauss transform~\cite{JL84}
to reduce $d$ to $d = O(\log n)$, and this only increases the distortion bound by a constant factor, which we can afford.
This way, a bound of $2^{d / t}$ would lead to $n^{1 / t}$ (for $t \geq 1$).
More specifically, such a transform $\phi : \mathbb{R}^d \to \mathbb{R}^m$ for $m = O(\log n)$ is data-oblivious, and can compute for each input point $x \in \mathbb{R}^d$ its image $\phi(x)$ in $O(d m)$ time and space,
which can be trivially implemented in dynamic or MPC settings.

\subsection{Proof of \Cref{thm:dynamic}: Dynamic Implementation}
\label{sec:proof_thm_dynamic}

First, recall that Johnson-Lindenstrauss transform~\cite{JL84} allows us to assume that $d = O(\log n)$ at a cost of an additive $\tilde{O}(d)$ term in update time and $O(1)$ in distortion.
Note that this is term is near linear with respect to the update size.

Our dynamic algorithm will maintain the tree embedding which is the output of the static \Cref{alg:frt}.
Hence, the main idea of the dynamic algorithm is to correctly maintain the labels $\ell_p^{(i)}$ for all points $p \in P$, for all levels $i \in [m]$. By \Cref{lemma:local_frt_distortion} the tree embedding defined by the labels will then have distortion $O(\Gamma \cdot \log \Gamma \cdot \log n)$, implying the correctness of the dynamic algorithm.

Before we present the update procedures, we introduce several auxiliary pieces of information that we maintain throughout our algorithm.

\paragraph{Auxiliary Variables.}

Throughout our algorithm, we compute and maintain the following auxiliary information:

\begin{enumerate}
    \item \label{item1_dyn} For all $p \in P$, we calculate $\pi(p)$ drawn uniformly at random from $[0,1]$.
    \item \label{item2_dyn} For all $p \in P$ and $i \in [m]$, we calculate the hash value $\varphi_i(p)$.
    \item \label{item3_dyn} For all $i \in [m]$ and $x \in \{\varphi_i(x') \mid x' \in \RR^d\}$ such that $\varphi^{-1}(x) \cap P \neq \emptyset$, we maintain $\varphi_i^{-1}(x) \cap P$ and $\pi_{\min}(\varphi_i^{-1}(x) \cap P)$.
\end{enumerate}

\Cref{item3_dyn} essentially says that on all levels $i \in [m]$ for any hash value $x$ with at least one point assigned to it from $P$ by $\varphi_i$ we maintain all points in $P$ with $\varphi_i$ value $x$. 
We note that we use standard dynamic set representations to maintain $\varphi_i^{-1}(x) \cap P$, and use min-heaps to maintain values $\pi_{\min}(\varphi_i^{-1}(x) \cap P)$.

We now describe the dynamic updates, namely point insertions and deletions. We note that, after every $\Theta(n)$ updates, we rebuild the entire structure by reinserting all current input points using the insertion procedure described below.

\paragraph{Point Deletions.} Whenever a point deletion occurs we might just ignore it and declare the leaf associated with the deleted point as inactive. This way we obtain a valid tree embedding of all points inserted into the input since the last rebuild.

\paragraph{Point Insertions.}
Now we turn to the more challenging task of point insertions.
Under point insertions, our dynamic algorithm will exactly maintain the output of the static \Cref{alg:frt}. As mentioned, this only requires us to correctly maintain the labels $\ell_p^{(i)}$ for all $p \in P, i \in [m]$. The details of the algorithm are presented in \Cref{alg:dyn_insertion}.

\begin{algorithm}[H]
\DontPrintSemicolon
\caption{\textsc{Insertion-procedure}$(p)$}
\label{alg:dyn_insertion}

\tcc{$\varphi_i$ is metric hashing with diameter bound $\tau_i := w_i/2$}

Draw $\pi(p)$ and calculate $\varphi_i(p)$ for $i \in [m]$\;

\For{$i \gets 1$ \KwTo $m$}{
    Enumerate hash values $X_i \gets \{\varphi_i(x) \mid x \in B(p, r_i)\}$\;
    
    \tcc{Calculating label $\ell_p^{(i)} = \pi_{\min}(\tilde B_i^P(p, r_i))$}

    Update $\varphi_i^{-1}(x)$ and $\pi_{\min}(\varphi_i^{-1}(x) \cap P)$ for all $x \in X_i$\;
    
    $\ell_p^{(i)} \gets \pi_{\min}\!\left(\bigcup_{x \in X_i} \tilde B_i^P(p, r_i)\right)$\;
    
    \tcc{Updating labels $\ell_q^{(i)}$ for $q \in P$}
    \If{$\pi(p) = \pi_{\min}(\varphi_i^{-1}(\varphi(p)))$}{
        \For {\textnormal{all} $x \in X_i$ \textnormal{and all} $p' \in \varphi_i^{-1}(x) \cap P$}{
            
                \If{$\pi(x) < \ell_{p'}^{(i)}$}{
                    Update $\ell_{p'}^{(i)} \gets \pi(p)$\;
                }
            
        }
    }
}
\end{algorithm}

At insertion of point $p \in \RR^d$, we first draw $\pi(p)$ and calculate $\varphi_i(p) $ for $ i \in [m]$, i.e. the auxiliary variables in \Cref{item1_dyn} and \Cref{item2_dyn}.

At each scale $i$, the algorithm \Cref{alg:dyn_insertion} then proceeds in two steps. First, we find the labels of the inserted point $p$ (lines 4-5), and then we update the labels of other points if necessary (lines 6-9).

For the first part, to calculate labels $\ell_p^{(i)}=\pi_{\min}(\tilde B_i^P(p,r_i))$ for the inserted point, we proceed as follows. For a specific $i \in [m]$, we first enumerate hash values $X_i = \{\varphi_i(x) \mid x \in B(p,r_i)\}$. Then, as for each $x \in X_i$ we are directly maintaining $\pi_{\min}(\varphi_i^{-1}(x))$, we can directly find $\ell_p^{(i)}$ (line 5).

In the second part, we will ensure that the labels of previously inserted points are correctly maintained. For any $i \in [m]$, we first check whether $\pi(p) = \pi_{\min}(\varphi_{i}^{-1}(\varphi(p))$, i.e. whether the inserted point $p$ has the smallest $\pi$ value in its corresponding bucket. If for some $i$ this is indeed the case, then the algorithm iterates through all hash values $x \in X_i$ and all points $p' \in \varphi_i^{-1}(x) \cap P$ associated with these hash values. If $\pi(x) < l_{p'}^{(i)}$, we update the label of $p'$ on level $i$ to $\pi(p)$ (line 9). 

\paragraph{Run Time Analysis.}

Now we discuss the run-time guarantees of the algorithm.

\begin{lemma}[Running time]
    \Cref{alg:dyn_insertion} runs in $\tilde{O}(\Lambda + d)$ amortized update time.
\end{lemma}

\begin{proof}
    We discuss insertions first. Let $p$ be the inserted point. By the theorem assumptions, we require $O(\poly d)=\tilde{O}(1)$ time to compute $\varphi(p)$. For calculating its labels $\ell_p^{(i)}$ for $i\in[m]$ (line 4-5), we first consider variables $\varphi^{-1}(x)$ and $\pi_{\min}(\varphi^{-1}(x)\cap P)$.  
    As there are only at most $\tilde{O}(n)$ hash values $x\in X_i$ over all levels $i$, $\varphi^{-1}(x)$ can be maintained in $\tilde{O}(1)$ time. For each non-empty hash value $x$, the maintenance of $\pi_{\min}(\varphi^{-1}(x)\cap P)$ can be done in $O(1)$ update time using min-heaps.
    For a specific $i \in [m]$, by the assumptions of the theorem on $\varphi$, we can enumerate hash values $X_i = \{\varphi_i(x) \mid x \in B(p,r_i)\}$ in $\tilde{O}(|X_i|)$ time which is in expectation $\tilde{O}(\Lambda)$. Finally, as for each $x \in X_i$, we are directly maintaining $\pi_{\min}(\varphi_i^{-1}(x))$, we can find $\ell_p^{(i)}$ in expected $O(\Lambda)$ time.

    Now we consider the updating the labels of existing points (lines 6-9). We first argue that we may charge the cost of changing $\ell_p^{(i)}$ this process to the points visited by it.

\begin{lemma}
\label{lemma:criticalvaluerecourse}

Let $S_i$ for $i \in [n]$ be i.i.d. samples drawn from a uniform $[0,1]$ distribution. Then the series defined by $\min\{S_i \mid i \in [k]\}$ for $k \in [n]$ contains $O(\log n)$ different values in expectation.

\end{lemma}

\begin{proof}

Let $I_k$ stand for the indicator variable of the event where $S_k = \min\{S_i \mid i \leq k\}$. As all $S_i$ are i.i.d. uniform $[0,1]$, we know that $\E[I_k] = 1/k$. Observe that  $\min_{i \in [k]}\{S_i\} \neq \min_{i \in [k-1}\{S_i\}$ if and only if $I_k = 1$. Hence, the number of different values the series $\min\{S_i \mid \in [k]\} | k \in [n]$ takes in expectation is $\E[\sum_{k=1}^n I_k] = \sum_{k= 1}^n \E[I_k] = O(\log n)$.
\end{proof}

Note that visiting one specific point takes $O(1)$ time. Consider when point $p \in P$ might be visited. This occurs when for some hash value $x \in \varphi_i(P \cap B(p,r_i))$ the value $\pi_{\min} (\varphi^{-1}_i(x) \cap P)$ is populated (in case it is the first point insertion associated with that hash value) or updated. By the definition of the hash function, we know that $\E[|x \in \varphi_i(P \cap B(p,r_i)) |] \leq \Lambda$ and by \Cref{lemma:criticalvaluerecourse}, each such hash value may change its minimal element according to $\pi$ in $P$ at most $O(\log n)$ times in expectation over $n$ point insertions. Hence, each point is visited at most $\tilde{O}(\Lambda)$ times in expectation over the series of $n$ updates. Combining all running-time components, we achieve the desired $\tilde{O}(\Lambda + d)$ amortized update time for insertion operations, with the additional $\tilde{O}(d)$ term resulting from the Johnson-Lindenstrauss transform.

The case of point deletions is now trivial. First, the deactivation of a node takes $O(1)$ time. Since we re-computing the embedding every $\Theta(n)$ updates by inserting all existing points into an empty input, we ensure that at all times the input only contains $O(n)$ inactive points. 
 Since the total update time of the algorithm over a sequence of $n$ points insertions is $\tilde{O}((d +\Lambda)\cdot n)$, the cost of these recomputation amortized over the $\Theta(n)$ updates between rebuilds works out to be an additive $\tilde{O}(d+\Lambda)$ factor, as desired.
    
\end{proof}

\paragraph{Queries.}

As discussed, our algorithm implicitly maintains the tree embedding by storing labels $\ell_p^{(i)}$ for all points $p \in P$ at all levels $i \in [m]$. Hence, the tree embedding can be queried by simply reporting these labels for each point. However, it is useful to examine more carefully the structural changes to the tree embedding that arise from label updates. This is especially relevant to applications of our dynamic tree embedding, as presented in \Cref{sec:application}.

\begin{lemma}[Recourse]
    The above dynamic algorithm can be extended to report the  solution after each update with $\tilde{O}(1)$ changes of the following types:

        \begin{itemize}
        \item Type 1): A leaf of the embedding becomes inactive in the sense that no point in $P$ corresponds to it.
        \item Type 2): A new leaf and a path connecting the leaf to an existing node in the tree embedding are inserted into the embedding.
    \end{itemize}
\end{lemma}

\begin{proof}

First, note that point deletions directly produce changes of Type 1.

Next, consider point insertions, which affect the embedding itself. Let $p$ be the inserted point.
We distinguish two parts: first, computing the labels $\ell_p^{(i)}$ of $p$, and then changing the labels $\ell_{p'}^{(i)}$ of existing points $p' \in P$.

When computing the labels of $p$ (lines 4–5 of \Cref{alg:dyn_insertion}), notice that at each level $i$, the inserted point either receives a new label that no other point at level $i$ has, or a label that some other points already have. Hence, exactly one new leaf and its connecting path may appear during each insertion; that is, an update of Type 2) occurs in this part.

On the other hand, changing the labels of existing points at level $i$ (lines 6–9 of \Cref{alg:dyn_insertion}) corresponds to certain subsets of points on some of the layers $i \in [m]$ changing their labels to correspond to the newly inserted point. We We will now show that this kind of update, where nodes switch parents within the tree, can be simulated by Type 1) and Type 2) updates. 

Suppose a node $v$ changes its parent in the tree from $u$ to $u'$. We can simulate this by first deactivating all leaves in the rooted subtree of $v$ (as if their corresponding points were deleted), and then re-inserting all the leaves of the same subtree—preserving its internal structure—under a new node $v'$ that is connected to $u'$.
Although all nodes of the subtree of $v$ remain present in the graph, this does not affect the tree distance metric, since the subtree contains no active leaves and therefore no minimum-length paths between active leaves pass through it.

It remains to argue that the simulation of edge changes through point deactivations and insertions does not lead to a significant blowup in update time or tree size. Consider the label of any point on any level $i \in [m]$. It will always correspond to some $\pi_{\min}$ value for some incremental set of at most $n$ uniform $[0,1]$ random variables. Hence, by \Cref{lemma:criticalvaluerecourse} it will be updated $O(\log n)$ times in expectation over the series of $n$ updates.

This implies that any leaf to root path may undergo at most $\tilde{O}(1)$ updates in expectation over a series of $n$ updates to the underlying points. Hence, if edge changes are simulated through leaf deactivations and insertions as described in expectation each leaf will be copied at most $\tilde{O}(1)$ times throughout the series of $n$ updates. 
\end{proof}

\subsection{Proof of \Cref{thm:mpc}: MPC Implementation}
\label{sec:proof_thm_mpc}

\thmmpc*

\paragraph{MPC Model.}
Before we proceed to proof, we review the MPC model~\cite{KarloffSV10}, specifically in our Euclidean setting which was also considered in e.g.~\cite{cohen2021parallel,DBLP:conf/icalp/CzumajGJK024,CzumajG0J25}.
The input is a subset of $[\Delta]^d$, and is distributed across machines each of $s = \Omega(\poly(d\log \Delta))$ memory.
The computation happens in rounds, and in each round, every pair of machines can communicate, subject to the constraint that the total bits of messages sent and received for every machine is $O(s)$.
We note that it takes $O(d \log \Delta)$ bits to represent a single point in $[\Delta]^d$,
and it is typical to assume $\Delta = \poly(n)$.
Moreover, as we mentioned, using a standard Johnson-Lindenstrauss transform~\cite{JL84},
$d$ can be made $O(\log n)$ without loss of generality (so $\poly(d) = \poly(\log n)$).
There is also a related formulation considered in~\cite{AhanchiAHKZ23}, where they focus on $(nd)^\epsilon$ local space.
Our result can also work in their setting, by reducing $d = O(\log n)$ using their same alternative version of JL transform~\cite{AhanchiAHKZ23},
which can be done with $(nd)^\epsilon$ local space requirement.

\begin{proof}[Proof of \Cref{thm:mpc}]
    
We only need to show how to simulate \Cref{alg:frt} in MPC, and the distortion bound follows immediately from \Cref{lemma:local_frt_distortion}.

For steps other than Line~\ref{line:kpi}, the value $\beta$ may be generated on a leader machine and then use a standard broadcast (see e.g.,~\cite{mpc_book}) to make it available to every machine,
the $\pi$ value may be generated independently on every machine.
The for-loop may be run in parallel, and $r_i$ may be computed locally.
The hash is deterministic and data-oblivious, so every machine can have the same hash without communication,
and within the space budget by \Cref{lemma:hashing}.

For Line~\ref{line:kpi}, our algorithm is almost the same to a geometric aggregation algorithm suggested in~\cite{DBLP:conf/icalp/CzumajGJK024},
albeit we cannot use its statement in a black-box way.
One first compute the min $\pi$-value for each bucket in parallel, via standard procedures sorting~\cite{GoodrichSZ11}, converge-casting and broadcasting in MPC (see e.g.~\cite{mpc_book}), all run in $O(\log_s n)$ rounds.
Now, since the hashing is data-oblivious, one can compute $\varphi_i(B(p, r_i))$ locally, and then again use standard MPC procedures to aggregate the $\ell_p^{(i)}$. This step increases the total space by a factor of $\Lambda$ which we can afford.
In total, the entire process runs in $O(\log_s n)$ rounds, and works for any $s = \Omega(\poly(d \log \Delta))$.
\end{proof}

 \section{Applications: New Dynamic Algorithms in High Dimension}
\label{sec:application}

In this section, we present dynamic algorithms for the following four problems:
\begin{itemize}
    \item \textbf{$k$-Median.} Given a point set $P\subset \mathbb{R}^d$, find a subset $S\subset P$ with $|S|=k$, such that the objective $\sum_{p\in P} \dist(p, S)$ is minimized. 
    \item \textbf{Euclidean bipartite matchings.} Given two point sets $A, B\subset \mathbb{R}^d$ with $|A|=|B|$, find a bijection $\mu: A\rightarrow B$ such that the objective $\sum_{a\in A} \dist (a, \mu(a))$ is minimized.
    \item \textbf{Euclidean general matchings. } Given point set $A \subset \RR^d$, find a partition $A=\tilde B_1\cup \tilde B_2$ and a bijection $\mu:\tilde B_1\rightarrow \tilde B_2$ such that the objective $\sum_{a\in \tilde B_1} \dist(a, \mu (a))$ is minimized. 
    \item \textbf{Geometric transport. } Let \(A, B\subset \mathbb{R}^d\) be any two sets of points such that each point \(a\in A\) has a non-negative integer \emph{supply} \(s_a\geq 0\) and each \(b\in B\) has a non-positive integer \emph{demand} \(d_b\leq 0\) satisfying $\sum_{a\in A} s_a  + \sum_{b\in B} d_b = 0.$ 
    Find an assignment \(\gamma: A \times B \rightarrow \mathbb{Z}_{\geq 0}\), satisfying $ s_a  = \sum_{b\in B} \gamma (a, b),  \forall a \in A$ and $-d_b  = \sum_{a\in A} \gamma (a, b), \forall b \in B $, such that the objective \(\sum_{a\in A,b\in B} \gamma(a, b)\|a-b\|_2\) is minimized.
\end{itemize}

The framework we present in this section allows maintenance the above objectives w.r.t. a tree metric $\mathcal{T}$, under a wide range of updates tree $T$ undergoes, namely:

\begin{itemize}
        \item \textit{Type 1)}: A leaf of the embedding becomes inactive in a sense that no point in $P$ corresponds to it anymore.
        \item \textit{Type 2)}: A new leaf and a path connecting the leaf to an existing point of the tree are inserted into the embedding.
\end{itemize}

Combining these results with the dynamic algorithm for tree embedding (\Cref{thm:dynamic}), we obtain dynamic algorithms for $k$-median (\Cref{thm:dyn_kmedian}), bipartite matching (\Cref{thm:EMD}), geometric transport (\cref{thm:geometrictransport}) and general matching (\cref{thm:generalmatching})  w.r.t. Euclidean metric.

Throughout this section, we will be referring to the embedding tree as $T$ with $m \in \tilde{O}(1)$ layers. To stay consistent with the notations of \cref{sec:local_frt}, the root will be on level $1$ and leaves will be on level $m$. For sake of simplicity, we will assume that the distance of a leaf and its parent is $1$, and hence the distance of a node on level $l$ and $l+1$ is $2^{m-l}$. We will refer to the level of point $p$ in the tree as $\ell(v) \in [m]$. We will refer to the set of input points represented by the leaves in the sub-tree of node $v$ as $P(v)$. For any node $v \in T$, we will refer to the set of nodes in the sub-tree of $v$ as $T_v$.

\subsection{$k$-Median}

In this section, we describe how to maintain $k$-median for a dynamic point-set $P$. Specifically, the main theorem of this section is the following. 

\begin{corollary}
\label{thm:dyn_kmedian}
    Given a dynamic point-set $P\subset \mathbb{R}^d$ and an integer $\alpha>0$, there is an algorithm that maintains an $O(\alpha^{3/2}\log \log \alpha \log n)$-approximation to the optimal $k$-median solution in $\tilde{O}(n^{1/\alpha} + d)$ update time. Given a point in $P$, its corresponding center can be reported in $\tilde{O}(1)$ time.
\end{corollary}

Specifically, we obtain $O_\epsilon(\log n)$-approximation with $\tilde{O}(n^\epsilon + d)$ update time for any $\epsilon > 0$ or $\tilde{O}(1)$-approximation with $\tilde{O}(d)$ update time. Note that each update consists of $\Omega(d)$ bits of information.

To prove the above result, we first present a dynamic algorithm that maintains an exact $k$-median w.r.t. tree metric $T$ under operations described in \Cref{sec:application}. \cref{thm:dyn_kmedian} is a direct implication of \cref{lemma:exact_k-median}, \cref{thm:dynamic} and \cref{lemma:hashing_dynamic}. 

\begin{lemma}
\label{lemma:exact_k-median}
    There exists a dynamic algorithm that maintains a solution to the $k$-median w.r.t. tree metric $\mathcal{T}=(T, \dist_T)$ under updates of Type 1) and 2), in $\tilde{O}(1)$ amortized update time.
\end{lemma}

Given node $v$ of $T$, let $T_v$ be the set of all leaves in subtree rooted at $v$. Given a tree $T$ and a set of centers $S$, we say that a leaf $v$ of $T$ is \emph{assigned a center at level $l$} if there is a leaf $u$ in $S$ such that the least common ancestor of $v$ and $u$ is at level $l$ and $u$ is the closest point to $v$ in $S$ on the tree. Throughout the algorithm it will always be the case that the centers are leaves. It is easy to prove (and was implicitly proven by \cite{cohen2021parallel}) that there is a minimum cost k-median solution where this is the case.

For level $l$ and leaf $v$, we define $\benf(v, l):=\sum_{i=m}^{l} (P_{i}(v)-P_{i+1}(v))\cdot (2^{m-l+1}-2^{m-i})$, where $P_{i}(v)$ is the number of leaves of subtree at level $i$ containing leaf $v$. This corresponds to how much we gain by opening $P(v)$ as a center, as opposed to the center the points of $P(v)$ are assigned to would be on at most level $l-1$. Now, at each node $v$ of level $\ell(v)$, we define $\benf(v):= \max_{u\in T_v} \benf(u, \ell(v))$, and the corresponding leaf $\kcenter(v):=\arg\max_{u\in T_v} \benf(u, \ell(v))$ maximizing this value. In case of equality between leaves we decide based on an arbitrary fixed ordering of the leaves.

\paragraph{The Static Algorithm of \cite{cohen2021parallel}.} 

For any leaf $u$ there can be multiple nodes $v \in T$ such that $\kcenter(v) = u$. By the definition of $\benf(v)$ these nodes must form a prefix of the $u$ root path. Furthermore, along this path by definition the $\benf(v)$ value is monotonically increasing as we walk towards the root. For each leaf $u$ let $\maxbenf(u)$ stand for the node $v \in T$ closest to the root such that $\kcenter(v) = u$. Define the list of pairs $L = \{(\benf(\maxbenf(u)), u) : u \in T_m\}$ ordered decreasingly based on the first value.

\cite{cohen2021parallel} has shown that the leafs corresponding to the first $k$ elements of $L$ form an optimal $k$-median solution. We will not repeat their proof but provide some intuition why this is the case.

The intuition behind values $\benf(v), \kcenter(v)$ is the following: assuming the $k$-median solution does not contain any leaf in the subtree $T_v$, but contains a leaf from the sub-tree of the ancestor of $v$ we may decrease the $k$-median cost by at most $\benf(v)$ through extending the solution with a leaf of $T_v$, and the leaf we should select for this purpose is $\kcenter(v)$. 

Let $L_{k-1}$ stand for the first $k-1$ elements of $L$ according to the ordering. Arguing inductively, assume that the set of leaves $S_{k-1}$ corresponding to $L_{k-1}$ is an optimal $k-1$-median solution. At this point, in order to obtain a $k$-median solution, we aim to find the leaf whose addition to the center set decreases the cost of the solution the most. We argue that this leaf should be the leaf corresponding to the first value of $L \setminus L_{k-1}$.

Consider any pair of values $(\benf(v), \kcenter(v))$ in $L \setminus L_{k-1}$. Set $S_{k-1}$ does not contain any leaf of $T_v$. If this would be the case, then along the path from $\kcenter(v)$ to $v$ there is a node $v’$ such that $\kcenter(v’) \neq \kcenter(v)$. However, all nodes $u$ for which $\kcenter(u) = \kcenter(v)$ lie along a path connecting $\kcenter(v)$ to $v$.

Furthermore, $S_{k-1}$ will contain a leaf from $T_u$, where $u$ is the ancestor of $v$. To see why this is the case, consider $\kcenter(u) \in T_u$ and the node $w$ closest to the root such that $\kcenter(u) = \kcenter(w)$. The pair $(\benf(w),\kcenter(u))$ is in $L$, and as benefit values are monotonically increasing along leaf to root paths we have $(\benf(w), \kcenter(u)) \in L_{k-1}$.   

This implies that for any leaf $u$, there is a pair $(\benf(v), u) \in L \setminus L_{k-1}$ satisfying that the addition of $u$ to $S_{k-1}$ would decrease the $k$-median cost by at most $\benf(v)$, and hence we should add the leaf corresponding to the first element of $L \setminus L_{k-1}$ to $S_{k-1}$ to obtain a $k$-median solution. 

Note that this argument is incomplete as it could be that, instead of aiming to extend a $k-1$-median solution, we should be looking for a structurally different $k$-median solution. For further details see \cite{cohen2021parallel}.

\paragraph{Our Dynamic Implementation.}

In order to obtain a dynamic $k$-median algorithm on $T$, it is sufficient to accurately maintain the ordered list $L$. Using standard data structures, it is straightforward to maintain the set of the first $k$ elements of $L$ in $\tilde{O}(1)$ update time per update to $L$.

Furthermore, in order to maintain $L$, it is sufficient to maintain $\benf(v), \kcenter(v)$ values for all $v \in T$, as there can be at most $\tilde{O}(1)$ nodes $v \in T$ such that $\kcenter(v) = u$ for any $u \in T_m$. Hence, we may maintain the list $L$ in $\tilde{O}(1)$ update time per update to $\benf(v), \kcenter(v)$ values, for all $v \in T$.

For each node $v \in T$, define the \emph{pessimistic estimate} $\pes(v):=P(v)\cdot 2^{m-\ell(v)}$, where we recall that $P(v)$ stands for the number of points in the subtree of $v$. This estimate corresponds to the cost of clustering the leaves of $T_v$ in the (hypothetical) situation that all points in this subtree are being assigned a center at level $l$.

We will first describe an alternative definition of $\benf(v)$ based on $\pes(v)$ which is easier to maintain in the dynamic setting.

\begin{claim}
\label{cl:k-median:benf}

For all non leaf $v \in T: \benf(v) = \pes(v) + \max_{u \in C(v)} \benf(u) - P_{\ell(v)+1}(u) \cdot 2^{m-\ell(v)}$. 

\end{claim}

\begin{proof}

\begin{align*}
    \benf(v) & =  \max_{u \in T_v} \benf(u, \ell(v)) \nonumber \\
    & = \max_{u \in T_v} \sum_{i = m}^{\ell(v)} (P_i(u) - P_{i+1}(u)) \cdot (2^{m-\ell(v)+1} - 2^{m-i}) \nonumber \\
    & = \max_{u \in T_v} \sum_{i = m}^{\ell(v) - 1} (P_i(u) - P_{i+1}(u)) \cdot (2^{m-\ell(v)+1} - 2^{m-i}) + (P_{\ell(v)}(u) - P_{\ell(v)+1}(u)) \cdot (2^{m-\ell(v)+1} - 2^{m - \ell(v)}) \nonumber \\
    & = \max_{u \in T_v} \sum_{i = m}^{\ell(v) - 1} (P_i(u) - P_{i+1}(u)) \cdot (2^{m-\ell(v)+1} - 2^{m-i}) + (P(v)- P_{\ell(v)+1}(u)) \cdot 2^{m-\ell(v)} \nonumber \\
    & = \pes(v) + \max_{u \in C(v)} \max_{\text{leaf } u' \in T_u}  \sum_{i = m}^{\ell(v) - 1} (P_i(u) - P_{i+1}(u)) \cdot (2^{m-\ell(v)+1} - 2^{m-i}) - P_{\ell(v)+1}(u) \cdot 2^{m-\ell(v)} \nonumber \\
    & = \pes(v) + \max_{u \in C(v)} \benf(u) - P_{\ell(v) + 1}(u) \cdot 2^{m-\ell(v)}
\end{align*} 

\end{proof}

Observe that a Type 1) or Type 2) update only changes $P(v)$ values for $v \in \Pi$ for some leaf to root path $\Pi$. Hence, values $P(v),\pes(v)$ can be maintained in $\tilde{O}(1)$ update time for all $v \in T$ throughout the update sequence. We will now explain how can we maintain $\benf(v)$ and $\kcenter(v)$ values for all $v \in T$ efficiently using these values.

In our dynamic implementation, each node \(v\) of $T$ at level $\ell(v)$ maintains a max heap $H(v)$ with elements of form $(\kcenter(u), \kvalue(u))$ for all $u \in C(v)$, and $\kvalue(u)$ is defined as $\benf(u)-P_{\ell(v)+1}(u) \cdot 2^{m-\ell(v)}$, sorted based on the second value. Hence, if $(u,x)$ is the first element of $H(v)$ then by \cref{cl:k-median:benf} $\benf(v) = x + \pes(v)$ and $\kcenter(v) = u$.

\textbf{Leaf deletions.} We first describe how to maintain the max-heaps $H(v)$ and $\benf(v), \kcenter(v)$ values when leaf $v \in T$ is deactivated, that is under Type 2) updates. 

\begin{algorithm}[H] \DontPrintSemicolon \caption{$\textsc{Delete-Leaf}(T, v)$} Let $\Pi$ be a path from $v.parent$ to root $r$\; 
$\benf(v) \leftarrow 0$ \; 
\For{vertex $u$ on level $i$ on path $\Pi$}{ Let $u' \gets $ child of $u$ such that $u'\in \Pi$ \;
Update the element of $H(u)$ corresponding to $u'$ to $(\kcenter(u'), \benf(u') - P_{\ell(u) + 1}(u') \cdot 2^{m-\ell(u)})$ \; $(w,x)\gets H(u).\max$ \; 
$\benf(u) \gets x + \pes(u)$ \; 
$\kcenter(u) \leftarrow w$ } \end{algorithm}

\textbf{Path insertions.} Now we describe the operations required to maintain the same values under Type 1) updates, that is when a path $\Pi$ is inserted connecting leaf $v$ to some already existing node $r_v$. Note that the process is split to two parts only because for the new set of nodes and $r_v$ a new element has to be added to their max-heaps.

\begin{algorithm}[H] \DontPrintSemicolon \caption{$\textsc{Insert-Path}(T, \Pi, v)$} $\kcenter(v) \leftarrow v$ \; 
$\benf(v) \leftarrow 1$ \; 
\For{vertex $u \neq v$ on path $\Pi$}{ Let $u' \gets u.\child$ ($u$ has a single child node at this point) \; 
Add element $(u', \benf(u') - 2^{m-\ell(u)})$ to $H(u)$ \; 
$\kcenter(u) \leftarrow v$ \; 
$\benf(u) \leftarrow \benf(u') - 2^{m-\ell(u)} + \pes(u)$ } Let $\Pi'$ be a path from $r_v$ to root $r$\; 
\For{vertex $u \neq r_v$ on path $\Pi'$}{ Let $u' \gets u.\child$ such that $u'\in \Pi \cup \Pi'$ \; 
Update the element $(u', \benf(u') - 2^{m-\ell(u)})$ corresponding to $u'$ in $H(u)$ $(w,x)\gets H(u).\max$ \;
$\benf(u) \gets x + \pes(u)$ \; 
$\kcenter(u) \leftarrow w$ } \end{algorithm}

\begin{claim}
    The algorithm correctly maintains values $\benf(v), \kcenter(v)$ for all $v \in T$ in $\tilde{O}(1)$ update time per Type 1) and Type 2) update to $T$. The algorithm answers center queries correctly in $\tilde{O}(1)$ time.
\end{claim}

\begin{proof}

    First, observe that when leaf $v$ is either inserted or deleted, $\kcenter$ and $\benf$ values may only change for nodes of $T$ whose sub-tree contains $v$, that is along a path $\Pi$ from $v$ to the root. For leaf $v$, the algorithm updates these values trivially. For any non-leaf $u \in \Pi$, the algorithm in a bottom up manner first restores the correctness of heap $H(u)$ (whose only value which has to be updated is the one corresponding to the single node in $C(u) \cap \Pi$). Afterwards, it updates $\benf(u)$ and $\kcenter(u)$ correctly due to \cref{cl:k-median:benf}.

    As this process only requires the algorithm to traverse a leaf root path of length $\tilde{O}(1)$, and at each leaf make a single update to a maximum heap value, this takes $\tilde{O}(1)$ worst-case time.
\end{proof}

\textbf{Queries.} In order to answer queries, its sufficient to maintain for all $v \in T$ a list of all leafs in $T_v$ which are part of the $k$-median solution. This information can be updated in $\tilde{O}(1)$ time per update to the $k$-median solution through walking along a leaf-root path.

When querying the closest center to a leaf $u$ in the $k$-median, we can just walk along the $u$-root path until we find some node $v \in T$ for which this list is not empty, and return an arbitrary leaf from the list.

\subsection{Bipartite Euclidean  Matching}
\label{sec:eucledianmatching}

This section is devoted to proving the following theorem concerning the dynamic bipartite Euclidean matching problem:

\begin{corollary}
\label{thm:EMD}

Let $A,B$, $|A| = |B| \leq n$ be sets of points in $\RR^d$ undergoing point pair insertions and deletions. There exists a dynamic algorithm maintaining an expected $O(\alpha^{3/2} \cdot \log \alpha \cdot \log n)$-approximate minimum cost Euclidean matching of $A$ to $B$ in $\tilde{O}(n^{1/\alpha} +d)$ expected amortized update time. 

\end{corollary}

In particular, our algorithm achieves $O_\epsilon(\log n)$ approximation in $\tilde{O}(n^\epsilon + d)$ update time for any constant $\epsilon>1$ or $\tilde{O}(1)$-approximation in $\tilde{O}(d)$ update time. Similarly to the $k$-median problem, using the dynamic tree embedding algorithm of \Cref{thm:dynamic}, in order to obtain a dynamic bipartite Euclidean matching algorithm, it is sufficient to solve the problem on dynamic 2-HST-s under the specific updates defined by the theorem (Type 1) and 2)). \cref{thm:EMD} is a direct implication of \cref{thm:dynamic}, \cref{lemma:hashing_dynamic} and \cref{lem:EMD:tree}.

\begin{lemma}

\label{lem:EMD:tree}

There exists a deterministic dynamic algorithm that maintains an optimum cost Euclidean matching of input points $A,B$ embedded into $T$ with respect to the tree metric $\mathcal{\tau} = (T,\dist_T)$ under Type 1) and Type 2) updates in $\tilde{O}(1)$ worst-case update time.

\end{lemma}

Define $d_T(a,b)$ to be the distance of leaves $a,b$ of 2-HST $T$ with respect to $T$, that is the tree distance metric of $T$. The bipartite Euclidean matching problem with respect to a tree embedding $T$ on sets of red and blue leaves $A,B$, corresponding red and blue points in $\RR^d$, requires us to find a matching (or bijection) $\mu(A) \rightarrow B$ minimizing $c(\mu) = \sum_{a \in A}d_T(a,\mu(a))$.

\textbf{Static Algorithm.}
We will first describe a simple static algorithm computing such a matching in time $\tilde{O}(n)$. Our dynamic implementation will exactly maintain the output of this algorithm. For all nodes $v \in T$, let $A(v)$ and $B(v)$ stand for the set of red and blue leaves in the sub-tree of $v$. Then $P(v)=A(v)\cup B(v)$ stands for the points associated with these leaves.

Iterate through the levels of the tree starting from the leaves and ending at the root. To each node $v$ of the tree, we assign a set of points $L(v)$. For any leaf $v$, the set $L(v)$ just contains the single point associated with them. For any non-leaf node $v$, the set $L(v)$ will consist of the points passed on to them by the child nodes.

Let $L_A(v)$ and $L_B(v)$ stand for the red and blue points in $L(v)$. Select a set of arbitrary $\min(|L_A(v),L_B(v)|)$ red and blue points $L'(v)$. Without loss of generality assume $|L_A(v) \geq L_B(v)|$. We assign a set $L'(v) \subseteq L(v)$ to be passed on from $v$ to its parent of size $|L_A(v)| - |L_B(v)|$ points of $A \cap L(v)$. 

Note that, as $|A| = |B|$, this implies that the root would not need to pass any points upwards, and for any point $p \in P$ there is a unique $v \in V$ such that $p \in L(v) \setminus L'(v)$. Furthermore, $|L'(v)| = ||A(v)| - |B(v)||$ and $L(v) \setminus L'(v)$ contains the same number of red and blue points. Then, obtain a matching $\mu(v)$ at node $v\in T$ by computing an arbitrary matching of red and blue points in $L(v) \setminus L'(v)$. Finally, the algorithm outputs the union matchings $\mu(v)$, that is $\mu = \cup_{v\in T} \mu(v)$.

\begin{algorithm}[H] \DontPrintSemicolon \caption{Static-Algorithm$(A, B)$} \label{alg:bipartite} \For{all leaves $v$ of $T$}{ $L'(v)=1$ \; 
If{$v$ is red}{ $L_A(v)=1$ } } \For{$i=m-1$ to $1$}{ \For{all nodes on level $i$ of $T$}{ $L_A(v) \gets \cup_{u\in C(v)}\{L'(v): L'(v) \text{ is red}\}$ \;
$L_B(v) \gets \cup_{u\in C(v)}\{L'(v): L'(v) \text{ is blue}\}$ \; $L(v) \gets L_A(v)\cup L_B(v)$ \; 
$L'(v) \gets$ monochromatic subset of size $||L_A(v)|-|L_B(v)||$ \; 
$\mu(v) \gets$ arbitrarily match the points in $L(v)\setminus L'(v)$ \; } } \textbf{return} $\cup_{v\in T} \mu(v)$ \end{algorithm}

Now we prove the correctness of the algorithm by first showing the following invariant.

\begin{lemma}[Invariant]
    \label{invariant:EMD}
    Matching $\mu$ satisfies that for all nodes $v \in T$ it matches all but $|A(v)-B(v)|$ uni-colored points of $P(v)$ with each other.
\end{lemma}

\begin{proof}
    First, by our algorithm, the number of matched points in any subtree rooted at $v$ is at most $\min \{|A(v)|, |B(v)|\}$, since until node $v$ we are only matching points inside $P(v)$. Second, for given $v\in T$, the points not matched in subtree rooted at $v$ are monochromatic (line 8).  Hence, it is enough to prove that the number of matched points in $T_v$ is at least $||A(v)|-|B(v)||$. 

    Suppose for contradiction that the algorithm matches less than $|A(v)-B(v)|$ points of $P(v)$ with each other at some node $v$ of $T$. That means that there is an unmatched red and a blue point in $P(v)$, denoted $a$ and $b$. Denote by $u$ the last common ancestor of leaves containing $a$ and $b$. By definition, $a\in L_A(u)$ and $b\in L_B(v)$. However, this means that both $a, b \in L'(v)$, a contradiction, since $L'(v)$ is monochromatic. Hence, at every node $v$ of $T$, all points but $|A(v)-B(v)|$ are matched in $P(v)$ with each other.
\end{proof}

\begin{lemma}[Correctness]
    \label{lemma:correctness:EMD}
    Any matching $\mu$ of points $P$ which satisfies \cref{invariant:EMD} is of optimal cost with respect to tree metric $\mathcal{\tau}$.
\end{lemma}

\begin{proof}
Let $\Pi$ stand for the multi-set of edges consisting of the union of the unique minimum length paths on the tree connecting $a$ and $\mu(a)$ for all $a \in A$. Consider an optimal cost matching $\mu^*$ of $A$ to $B$ on the tree. Let $\Pi^*$ stand for the multi-set of edges of the union of the unique minimum length paths on the tree connecting $a$ and $\mu^*$ of $A$ for all $a \in A$.

Fix an arbitrary non-root node $v \in T$ and its parent $u$. $\mu^*$ must match at least $|A(v)-B(v)|$ points of $P(v)$ with points not in $P$. Hence, $\Pi^*$ will contain edge $(u,v)$ at least $|A(v) - B(v)|$ times. Now consider $\Pi$. By the invariant it matches at most $|A(v)-B(v)|$ points of $P(v)$ with points not in $P(v)$. Hence, $\Pi$ contains edge $(u,v)$ exactly $|A(v)-B(v)|$ times. As edge weights are non-negative in $T$ this implies that $c(\mu) \leq c(\mu^*)$ (and as $c(\mu^*)$ is of minimum cost by assumption $c(\mu) = c(\mu^*)$).
\end{proof}

\paragraph{Dynamic Implementation.} In our implementation the output of this simple static algorithm can be maintained efficiently under the type of updates defined by \Cref{thm:dynamic}. First, observe that if we insert a path starting from a newly inserted leaf with a deactivated point to an existing node, in order to update the output of the algorithm, it is sufficient to set $L(v)$ and $L'(v)$ values to be empty for each inserted node $v$.

We will describe a constant time operation how, starting from a leaf and walking to the root, we may handle point deactivations and point activations (\Cref{alg:dyn_bipartite_insert_gen} and \Cref{alg:dyn_bipartite_delete_gen}).  We say that node $v$ needs to handle a point insertion (or deletion) if for some child node $u$ of $v$ the set $L'(u)$ undergoes a point insertion (or deletion). Both algorithms consist of procedures described in \Cref{alg:dyn_bipartite_insert} and \Cref{alg:dyn_bipartite_delete}, in a bottom-up manner.

\begin{algorithm}[H] \DontPrintSemicolon \caption{$\textsc{Insertion}(A, B, a)$} \label{alg:dyn_bipartite_insert_gen} $v_a \gets$ leaf corresponding to point $a$ \;
$\textsc{Insertion\_at\_node}(A, B, a, v_a)$ \end{algorithm}

\begin{algorithm}[H] \DontPrintSemicolon \caption{$\textsc{Deletion}(A, B, a)$} \label{alg:dyn_bipartite_delete_gen} $v_a \gets$ leaf corresponding to point $a$ \;
$\textsc{Deletion\_at\_node}(A, B, a, v_a)$ \end{algorithm}

First, assume, without loss of generality, that node $v \in T$ needs to handle a point insertion of $a \in A$ (see \Cref{alg:dyn_bipartite_insert}). First we need to extend $L(v)$ by $a$. If $A(v) \geq B(v)$ then before the insertion $L'(v)$ was either empty or consisted points of $A$. In this case we may simply extend $L'(v)$ by $a$ as well and ask the parent of $v$ to handle the insertion of $a$.

If $A(v) < B(v)$, then we will need to match $a$ to some arbitrary point of $b \in B \cap L'(v)$ which was previously passed on by $v$ to its ancestor. We set $L'(v) = L'(v) \setminus \{b\}$ and set $\mu(a) = b$ and $\mu(b)$ to be undefined if it previously had a value. At this point we ask the parent of $v$ to handle the deletion of point $b$. 

\begin{algorithm}[H] \DontPrintSemicolon \caption{$\textsc{Insertion\_at\_node}(A, B, a, v)$} \label{alg:dyn_bipartite_insert} \If{$v$ is root of $T$}{ \textbf{return} $c(\mu)$ \; } 
$L(v) \gets L(v) \cup \{a\}$ \; 
\If{$|A(v)| < |B(v)|$}{ $b \gets$ point in $B\cap L'(v)$ \; 
$L'(v) \gets L'(v) \setminus \{b\}$ \; 
$\mu(a)=b$ \; 
$c(\mu) \gets c(\mu)+d_T(a, b)$ \; 
$\textsc{Deletion\_at\_node}(A, B, b, v.parent)$ \; } 
\Else{ $L'(v) \gets L'(v)\cup \{a\}$ \; 
$\textsc{Insertion\_at\_node}(A, B, a, v.parent)$ \; } 
\end{algorithm}

We can describe an analogous process for the deletion of point $a$ from node $v$ (see \Cref{alg:dyn_bipartite_delete}). We first set $L(v) = L(v) \setminus \{a\}$. If $a \in L'(v)$ then we can simply set $L'(v) = L'(v) \setminus \{a\}$ and ask the parent of $v$ to handle the deletion of $a$. Otherwise, if $A(v) > B(v)$ before the update we can select an arbitrary $a'$ point of $L'(v)$ and set $\mu(a') = \mu(a)$, delete $a$ from $\mu$, set $L'(v) = L'(v) \setminus \{a'\}$ and ask the parent of $v$ to handle the deletion of $a'$. If $B(v) \geq A(v)$ before the deletion then let $b = \mu(a)$, we, delete the value of $a$ from $\mu$, set $L'(v) = L'(v) \cup b$ and ask the parent of $v$ to handle the insertion of $b$. 

Now we prove the correctness and the running time of these algorithms.

\begin{lemma}[Correctness]
\label{lemma:correctness_emd}
    Dynamic algorithms \Cref{alg:dyn_bipartite_insert_gen} and \Cref{alg:dyn_bipartite_delete_gen} maintain the solution to the Euclidean bipartite matching w.r.t. tree metric $\mathcal{T}$.
\end{lemma}

\begin{proof}
    We prove the above claim by showing that the invariant \Cref{invariant:EMD} is maintained after every update.
    
    Suppose $a$ is inserted into $A$, and let $v_a$ be the corresponding leaf containing $a$. Clearly, for vertices that are not on path $\Pi$ from $v_a$ to root $r$, the invariant holds, as the sets $A(v), B(v)$ do not change. For vertices $v$ on path $\Pi$, notice that the set $A(v)$ increases by one. In case $|A(v)| \geq |B(v)|$ before the update, by the algorithm the number of unmatched points from $P(v)$ also increases by 1, as does $||A(v)|-|B(v)||$, and hence the invariant holds for $v$. On the other hand, if $|A(v)| < |B(v)|$ before the update, the number of matched pairs among $P(v)$ increases by one. Hence, the number of unmatched points decreases by 1, as does $||A(v)|-|B(v)||$, so the invariant holds.
\end{proof}

\begin{algorithm}[H] \DontPrintSemicolon \caption{$\textsc{Deletion\_at\_node}(A, B, a, v)$} \label{alg:dyn_bipartite_delete} 
\If{$v$ is root of $T$}{ \textbf{return} $c(\mu)$ \; } 
$L(v) \gets L(v)\setminus \{a\}$ \;
$L(v) \gets L(v)$ \; 
\If{$|A(v)| \leq |B(v)|$}{ Let $b \gets$ blue point in $L_B(v)$ matched to $a$ \;
$c(\mu) \gets c(\mu) - d_T(a, b)$ \; 
$\textsc{Insertion\_at\_node}(A, B, b, v.parent)$ \; } 
\Else{ $L'(v) \gets L'(v)\setminus \{a\}$ \; 
$\textsc{Deletion\_at\_node}(A, B, a, v.parent)$ \; } 
\end{algorithm}

\begin{lemma}
\label{lemma:runtime_emd}
    The running time of the above dynamic algorithm is $\tilde{O}(1)$ worst-case.
\end{lemma}

\begin{proof}
Either of the processes follow a leaf to root path in the tree and take $O(1)$ time at each node along it, hence take $\tilde{O}(1)$ worst-case time to handle. 
\end{proof}

\Cref{thm:EMD} now follows directly from
\Cref{lemma:correctness_emd} and \Cref{lemma:runtime_emd} and \Cref{thm:dynamic}.

\subsection{Geometric Transport}

In this section we describe our implementation of the fully dynamic geometric transport problem. The algorithm is mostly identical to that of \cref{sec:eucledianmatching} for Euclidean bipartite matching. However, in the case of geometric transport we allow for point pairs with large weights to be inserted and deleted from the input, hence the recourse of the actual assignment describing the optimal solution can be large after each update. Therefore, our algorithm only maintains the cost of the optimal solution.

\begin{corollary}

\label{thm:geometrictransport}

Let $A,B$, $|A| = |B| \leq n$ with weights $w_a, w_b$ be sets of points in $\RR^d$ undergoing point pair insertions and deletions. There exists a dynamic algorithm maintaining the cost of an expected $O(\alpha^{3/2} \cdot \log \alpha \cdot \log n)$-approximate minimum cost geometric transport assignment of $A$ to $B$ in $\tilde{O}(n^{1/\alpha} +d)$ expected amortized update time. 

\end{corollary}

Similarly to our other applications, the algorithm can either obtain $O_\epsilon(\log n)$-approximation in $\tilde{O}(n^\epsilon + d)$ update time or $\tilde{O}(1)$-approximation in $\tilde{O}(d)$ update time. Our dynamic implementation will maintain the exact cost of an optimal solution with respect to the tree embedding described by \cref{thm:dynamic}. Hence, \cref{thm:geometrictransport} is a direct implication of \cref{lemma:hashing_dynamic} and \cref{lem:geometrictransport:tree}.

\begin{lemma}

\label{lem:geometrictransport:tree}

There exists a deterministic dynamic algorithm that maintains the optimal geometric transport cost of points $P$ with integer weights $w$ embedded into $T$ with respect to the tree metric $\mathcal{\tau} = (T,\dist_T)$ under Type 1) and Type 2) updates in $\tilde{O}(1)$ worst-case update time.

\end{lemma}

For node $v \in T$ let $A_w(v)$ and $B_w(v)$ stand for the multisets of red and blue points contained within its sub-tree. Let $\ell(v)$ stand for the level of node of $v$ in the tree where leaves stand at level $m$ and the root stands at level $0$. Assume that an edge between a leaf and its parent is of length $1$. This means that an edge between a node of level $l$ and its parent is of length $2^{m-l}$.

We have the following result about the cost of the geometric transport.

\begin{lemma}
    The value of the geometric transport cost w.r.t. tree metric $\mathcal{T}$ is 
    $$\tilde{\mu} = \sum_{v \in T}  ||A_w(v)|-|B_w(v)|| \cdot 2^{m-\ell(v)}.$$
\end{lemma}

\begin{proof}

First, notice that we can trivially turn an instance of geometric transport into an instance of Euclidean bipartite matching through creating $w(v)$ copies for all leaves $v$ in $T$. Therefore, it is enough to show that the cost of the matching $\mu$ produced by the static algorithm \Cref{alg:bipartite} of \cref{sec:eucledianmatching} is exactly $\tilde{\mu}=\sum_{v \in T}  ||A_w(v)|-|B_w(v)|| \cdot 2^{m-\ell(v)}$. 

Now fix some arbitrary optimal matching $\mu^*$ on this Euclidean matching instance. Let $\Pi$ and $\Pi^*$ correspond to the multi-set of edges described by the union of paths on $T$ connecting the matched points of $\mu$ and $\mu^*$ respectively.
Fix an arbitrary edge between $v$ and its parent $u$. In $\Pi^*$ this edge must appear at least $||A_w(v)|-|B_w(v)||$ as that many points of $P(v)$ cannot be matched to each other. Similarly, due to \cref{invariant:EMD} this edge will appear at most $||A_w(v)|-|B_w(v)||$ times.

The paths in $\Pi$ crossing this edge contribute at most $||A_w(v)|-|B_w(v)|| \cdot 2^{m-\ell(v)}$ to $c(\mu)$. Similarly, the paths in $\Pi^*$ crossing this edge due to crossing it contribute at least $||A_w(v)|-|B_w(v)|| \cdot 2^{m-\ell(v)}$ to $c_T(\mu^*)$ (as $\mu^*$ might not match $\min\{|A_w(v)|, |B_w(v)|\}$ pairs from $P(v)$). Summing over all $v\in T$, we get $c(\mu) \leq \tilde{\mu} \leq c(\mu^*)$. As $c(\mu) \geq c(\mu^*)$, we finally conclude that $c(\mu)=c(\mu^*)=\sum_{v \in T}  ||A_w(v)|-|B_w(v)|| \cdot 2^{m-\ell(v)}$.
\end{proof}

Due to any update to the input points $P$ at most $m = \tilde{O}(1)$ nodes of $T$ might change their $A_w$ or $B_w$ values. These nodes are also easy to found as they lie along the path connecting the leaf containing the updated point and the root. Therefore, it is straightforward to maintain $\tilde{\mu}$ in the dynamic setting in $\tilde{O}(1)$ worst-case update time per update to $P$.

\subsection{Euclidean Matching}
\label{sec:generalmatching}

In this subsection we describe our dynamic algorithm for dynamic Euclidean matching. The algorithm is similar to that of \cref{sec:eucledianmatching} for the bipartite version of the problem. 

\begin{corollary}
\label{thm:generalmatching}

Let $P$, $|P| \leq n$ be a set of points in $\RR^d$ undergoing point insertions and deletions. There exists a dynamic algorithm which maintains an expected $O(\alpha^{3/2} \cdot \log \alpha \cdot \log n)$-approximate minimum cost matching of $P$ in expected $\tilde{O}(n^{\alpha^{-1}} + d)$ amortized update time.
\end{corollary}

Specifically, the algorithm obtains $O_\epsilon( \log n)$-approximation in $\tilde{O}(n^{\epsilon} + d)$ update time or $\tilde{O}(1)$-approximation in $\tilde{O}(d)$ update time.

We will maintain a tree embedding of the input points using \cref{thm:dynamic}. Let $T$ be the $2$-HST maintained by the algorithm, where leaves correspond the points of $P$. We will maintain an exact solution to the problem with respect to the distance metric defined by $T$. Then \cref{thm:generalmatching} is a direct implication of \cref{lemma:hashing_dynamic} and \cref{lem:generalmatching:tree}.

\begin{lemma}

\label{lem:generalmatching:tree}

There exists a deterministic dynamic algorithm that maintains a minimum cost perfect matching of the active leaves of a tree embedding $T$ with respect to the tree metric $\mathcal{\tau} = (T,\dist_T)$ under Type 1) and Type 2) updates in $\tilde{O}(1)$ worst-case update time.

\end{lemma}

We will first define a simple static algorithm for the problem on $T$. Let $T_i$ stand for the $i$-th layer of vertices in the tree embedding from the root, that is $T_1$ stands for a set containing the root and $T_m$ stands for the set containing the leaves. Recall that $m = O(\log n)$. For non-leaf node $v \in T$ let $C(t)$ stand for the set of child nodes of $v$. For node $v \in T$ let $P(v)$ stand for the set of points represented by the leaves of the sub-tree rooted at $v$.

We will further define two sets for each node $v \in T$, $L(v)$ and $L'(v)$. For any leaf node $v \in T_m$ we initially set $L(v) = L'(v) = P(v)$ to consist of the single point stored in $v$. We initially set the output matching $\mu$ to be empty.

The static algorithm will iterate through the layers of the tree starting from $i=m-1$ ending at $1$. During these iterations for all $v \in T_i$ the algorithm first collects all the points forwarded to it by its children. If it receives an even number of points it matches all of them arbitrarily. If it receives an odd number of points it matches all but one of them and forwards an arbitrary point to its parent.

\begin{algorithm}[H] 
\DontPrintSemicolon \caption{$\textsc{Processing\_Node}(v)$} $L(v) \leftarrow \cup_{u \in C(u)} L'(v)$ \; 
\If{$|L(v)|$ is odd}{ $L'(v) \leftarrow \{p\}$ for some arbitrary point $p \in L(v)$ } 
\Else{ $L'(v) \leftarrow \emptyset$ } 
Extend $\mu$ with an arbitrary matching of $L(v) \setminus L'(v)$ 
\end{algorithm}

Observe that the algorithm visits each node $v \in T$ at most once where it takes $O(|L(v)|)$ time to handle its operations. Furthermore, each point of the tree may only appear in $L(v)$ sets along a leaf to root path, hence compiles in $\tilde{O}(n)$ time. We will now state an invariant which will be satisfied after a run of the static algorithm and remains satisfied after handling each update. This invariant will be sufficient to argue about the correctness of the algorithm.

\begin{invariant}

\label{invariant:generalmatching}

For all nodes $v \in T$ matching $\mu$ matches all points of $P(v)$ with each other if $|P(v)|$ is even, otherwise it matches all but one.

\end{invariant}

Observe that the matching $\mu$ computed by the static algorithm satisfies \cref{invariant:generalmatching} by construction. Now we will show that this condition is sufficient to prove optimality.

\begin{claim}

\label{cl:generalmatching:correctness}

If matching $\mu$ satisfies \cref{invariant:generalmatching} then $\mu$ is a minimum cost matching of $P$ with respect to the tree metric $\mathcal{\tau}$.

\end{claim}

\begin{proof}

Fix a minimum cost matching with respect to the tree metric $\mu^*$. Let $\Pi$ and $\Pi^*$ stand for the multi-set described by the edges of union of the paths between the matched points of $\mu$ and $\mu^*$ respectively.

Consider an arbitrary node $v \in T$ with parent $u$. If $|P(v)|$ is odd then $\mu^*$ must match a point of $P(v)$ with some point outside of $P(v)$. Hence $\Pi^*$ contains $(u,v)$. By assumption $\mu$ only matches one point of $P(v)$ outside of $P(v)$, hence $\Pi$ contains $(u,v)$ once. If $|P(v)|$ is even then, $\mu$ matched all points of $P(v)$ with each other, hence $(u,v) \notin \Pi$. This implies that $c_T(\mu) \leq c_T(\mu^*)$.
\end{proof}

\paragraph{Dynamic Implementation.} We will now describe how to dynamically maintain a matching $\mu$ satisfying \cref{invariant:generalmatching} under Type 1) and Type 2) updates in worst-case $\tilde{O}(1)$ update time. For sake of simplicity we will assume that $|P|$ is even at all times, that is updates come in pairs.

We will describe two types of operations: \textbf{Insert$(v,p)$} and \textbf{Delete$(v,p)$} for $v \in T, p \in P$.

If \textbf{Insert$(v,p)$} is called for node $v \in T$ with parent $u$ we complete the following steps:

\begin{algorithm}[H] \DontPrintSemicolon \caption{$\textsc{Insert}(v,p,u)$} $L(v) \leftarrow L(v) \cup \{p\}$ \; \If{$|L'(v)| = 1$}{ $p' \leftarrow L'(v)$ \; Extend $\mu$ with $p,p'$ \; $\textbf{Delete}(u,p)$ } \Else{ $L'(v) \leftarrow \{p\}$ \; $\textbf{Insert}(u,p')$ } \end{algorithm}

If \textbf{Delete$(v,p)$} is called for node $v \in T$ with parent $u$ we complete the following steps:

\begin{algorithm}[H] 
\DontPrintSemicolon \caption{$\textsc{Delete}(v,p,u)$} 
$L(v) \leftarrow L(v) \setminus \{p\}$ \; 
\If{$L'(v) = \{p\}$}{ $L'(v) \leftarrow \emptyset$ \; 
\textbf{Delete$(u,p)$} } 
\Else{ $p' \leftarrow \mu(p)$ \; 
\If{$L'(v) = \emptyset$}{
Remove $p,p'$ from $\mu$\; $L'(v) \leftarrow \{p'\}$ \; 
$\textbf{Insert}(u,p')$ } 
\Else{ $p'' \leftarrow L'(v)$ \; 
Remove $p,p'$ from $\mu$ \; 
Extend $\mu$ by $p',p''$ \; $\textbf{Delete}(u,p'')$ } 
} 
\end{algorithm}

If $v$ does not have a parent meaning it is the root of $T$ it repeats the same operations without the recursive call. Note that for sake of simplicity we assume that $|P|$ is even at all times that is updates come in pairs. At deletion or insertion of point $p$ associated with leaf $v$ the algorithm simply calls \textbf{Delete($v,p$)} or \textbf{Insert($v,p$)} respectively.

The recursion follows a leaf to root path and takes $O(1)$ time to execute for a given note hence runs in $O(m) = \tilde{O}(1)$ time. Furthermore, it ensures that from $P(v)$ only the up to one point in $L'(v)$ is not matched with points in $P(v)$. Hence, the algorithm satisfies \cref{invariant:generalmatching} and by \cref{cl:generalmatching:correctness} it maintains a minimum cost matching of $P$ with respect to the tree metric $\mathcal{\tau}$.

\appendix
\section{Local Queries for Hash Functions.}
\label{app:hashing_localsearch}

In \Cref{lemma:hashing_dynamic}, we state that the $(\Gamma, \Lambda)$-hash $\varphi$ with diameter bound $\tau$
can evaluate $\varphi(B(p, \tau/ \Gamma))$ in $O(\poly(d) \cdot |\varphi(B(p, \tau / \Gamma))|$ time.
However, this part of the statement did not appear in~\cite{filtser2025fasterapproximationalgorithmskcenter}, and here we provide a proof sketch.

We first review the construction of hashing $\varphi$ in~\cite{filtser2025fasterapproximationalgorithmskcenter}.
The hash is scale-invariant, and one can focus on $\tau = \sqrt{d}$.
The hash buckets are defined by unit hypercubes in $\mathbb{R}^d$, whose diameter is bounded by $\tau$,
shifted by a vector $v \in [0, 1]^d$ uniformly at random.
However, the shift is only relevant to bounding the expectation $\E[|\varphi(B(p, \tau / \Gamma))|]$,
and we do not need the randomness of $v$ for evaluating $\varphi(B(p, \tau / \Gamma))$;
in other words, our argument conditions on a fixed shift $v$.
Hence, without loss of generality, we can translate the coordinates so that $v = 0$, i.e.,
the buckets are of the form $\times_{i \in [d]} [t_i, t_{i + 1})^d$, over all $t_i \in \mathbb{Z}$, $i \in [d]$.

Our task of evaluating $\varphi(B(p, \tau / \Gamma))$ is equivalent to finding the intersecting unit hypercubes of $B(p, \tau / \Gamma)$.
Now, consider an auxiliary graph $G$, where the nodes are all the unit hypercubes.
Two nodes/hypercubes share an edge if the two hypercubes are adjacent and differ in only one coordinate.
Let $z := \varphi(p)$ be the hypercube/node that contains $p$.
Then all nodes that belong to $\varphi(B(p, \tau / \Gamma))$ form a connected component in $G$,
and that each node has degree $O(d)$.
Therefore, we can evaluate $\varphi(B(p, \tau / \Gamma))$ from the node $z$ via a DFS (depth-first search),
and we terminate whenever the current node does not belong to $\varphi(B(p, \tau / \Gamma))$.
We end up with visiting at most $O(d) \cdot |\varphi(B(p, \tau / \Gamma))|$ nodes during this process (the $O(d)$ comes from the degree bound),
which finishes the proof. \qed

\bibliography{ref}
\bibliographystyle{alphaurl}

\end{document}